\documentclass[11pt]{article}
\usepackage[utf8]{inputenc}
\usepackage{setspace}
\usepackage{bbm}
\usepackage{framed}

\usepackage{amsmath,amsfonts,amssymb,mathtools}
\usepackage{xcolor}
\usepackage[colorlinks=true,
            citecolor=blue]{hyperref}
\usepackage{fullpage}
\usepackage{amsthm}
\usepackage{indentfirst}
\usepackage{tikz}
\usepackage{xspace}
\usetikzlibrary{arrows.meta}
\allowdisplaybreaks
\newtheorem{theorem}{Theorem}
\newtheorem{lemma}{Lemma}
\newtheorem{claim}{Claim}
\newtheorem{corollary}{Corollary}
\theoremstyle{definition}
\newtheorem{definition}{Definition}
\usepackage{subcaption} %
\usepackage[capitalize]{cleveref}
\crefname{figure}{Figure}{Figures}
\Crefname{figure}{Figure}{Figures}
\crefname{claim}{Claim}{Claims}
\Crefname{claim}{Claim}{Claims}

\usepackage{todonotes}

\usepackage{palatino}

\newcommand{\forr}{\operatorname{forr}}
\newcommand{\gsp}{\operatorname{GSP}}
\newcommand{\tv}{\operatorname{TV}}

\newcommand{\yes}{\textsf{yes}\xspace}
\newcommand{\no}{\textsf{no}\xspace}

\newcommand{\F}{\ensuremath{\mathbb{F}}}

\providecommand{\email}[1]{\href{mailto:#1}{\nolinkurl{#1}\xspace}}

\title{The Quantumly Fast and the Classically Forrious\footnote{For a discussion of a previous version of this work claiming an analogous lower bound, see~\cref{app:previousversion}.}}

\author{
Cl\'ement L. Canonne\thanks{The University of Sydney.
Email: \email{clement.canonne@sydney.edu.au}}
\and Kenny Chen\thanks{The University of Sydney.
Email: \email{kche5493@uni.sydney.edu.au}} \and 
Juli\'an Mestre\thanks{The University of Sydney.
Email: \email{julian.mestre@sydney.edu.au}}
}

\begin{document}
\maketitle
\begin{abstract}
We study the \emph{extremal Forrelation problem}, where, provided with oracle access to Boolean functions $f$ and $g$ promised to satisfy either $\forr(f,g)=1$ or $\forr(f,g)=-1$, one must determine (with high probability) which of the two cases holds while performing as few oracle queries as possible. It is well known that this problem can be solved with \emph{one} quantum query; yet, Girish and Servedio (ITCS 2026) recently showed this problem requires \smash{$\widetilde\Omega(2^{n/4})$} classical queries, and conjectured the optimal bound to be \smash{$\widetilde\Theta(2^{n/2})$}. By generalizing their construction, we build on their result and prove a non-adaptive lower bound of $\Omega(2^{(1/2- o(1))n})$, which matches the conjectured lower bound up to a vanishing constant in the exponent.
\end{abstract}
\section{Introduction}\label{sec:introduction}

One of the first examples of provable quantum advantage is famously due to Shor~\cite{DBLP:journals/siamcomp/Shor97}, who showed factoring could be performed in polynomial time on a quantum computer. Ever since then, two main overarching questions have been considered. First, for what other problems are there provable quantum speedups? Secondly, how much more powerful are quantum computers than classical computers? In view of our dire lack of classical lower bounds in most computational models, which would be a prerequisite to establishing any strong separation between the classical and quantum settings, one natural way to analyze these two questions is in the \emph{query complexity model}. Here, an algorithm is given black box access to a function, and tasked to learn some property of the function. The question then asked is what the minimum number of queries needed by algorithms, both classical and quantum, is in order to perform this task.

Many quantum speedups have been obtained in this query complexity model, such as for searching \cite{DBLP:conf/stoc/Grover96} and property testing \cite{DBLP:journals/toc/MontanaroW16}, to name a few: the interested reader is referred to~\cite{hamoudi2025briefintroductionquantumquery} for a survey on quantum query complexity results. One of the flagship tasks in this query complexity model is the \emph{Forrelation} problem, introduced by Aaronson \cite{DBLP:conf/stoc/Aaronson10}. At a high level, this asks an algorithm to determine, in as few queries as possible, if a function $f$ is (significantly) correlated or anticorrelated with the Fourier transform of another function $g$. The reason for this problem garnering so much attention is that it shows a maximal separation between quantum and classical queries. Namely, Aaronson and Ambainis showed that a quantum computer can successfully perform this task with a \emph{single} query, yet any classical computer requires \emph{exponential} query complexity~--~specifically, $\widetilde\Omega(2^{n/2})$ queries~\cite{DBLP:journals/siamcomp/AaronsonA18}.\footnote{Throughout, the $\tilde{\Omega}$ notation omits polylogarithmic factors in the argument: i.e., in this case, $\mathrm{poly}(n)$ factors.} This lower bound holds when the (anti)correlation threshold of the promise problem is some $\varepsilon =\Omega(1)$, bounded away from 1, but the techniques employed fail as $\varepsilon$ gets close to 1 (specifically, they cannot get beyond $\varepsilon = 2/\pi \approx 0.64$), and in particular left open the \emph{extremal regime} ($\varepsilon = 1$), where $g$ is promised to either be perfectly correlated or anticorrelated with the Fourier transform of $f$. Interestingly, in this setting, the quantum algorithm is able to succeed \emph{with probability one}, revealing some insights about the power of exact quantum computation. 
Very recently, Girish and Servedio~\cite{GirishS26} studied this extremal version of the Forrelation problem (formally defined in \cref{def:extremal_Forrelation_problem}), and were able to establish a lower bound of $\widetilde\Omega(2^{n/4})$ classical queries; further, they conjectured the lower bound to be $\widetilde\Omega(2^{n/2})$, matching that obtained by Aaronson and Ambainis in the non-extremal case. 

Our main technical contribution is an almost tight lower bound for the extremal Forrelation problem, showing that Forrelation is, truly, \emph{``extremally hard''}:
\begin{theorem}[Main Result (Informal, see~\cref{thm:main})]
    \label{thm:main:forrelation}
    Any \emph{non-adaptive} classical algorithm for extremal Forrelation must make $\Omega(2^{(1/2-o(1))n})$ queries. (Moreover, $O(2^{n/2})$ queries suffice.)
\end{theorem}
\noindent Note that, as we discuss later, our result only applies to non-adaptive algorithms, and does not rule out a better query complexity for classical algorithms able to adaptively choose their queries. 

\paragraph{But why should we care?} The reader only tangentially interested in (near)-optimal lower bounds for the extremal version of influential quantum complexity promise problems may, at this point, be unconvinced by the appeal of our main theorem. While we respectfully disagree, we emphasize that there \emph{is} motivation beyond this maximal query complexity motivation! Taking a step back, the reason for which Girish and Servedio were able to push the lower bound beyond $\varepsilon=2/\pi$ was their use of an entirely different technique than Aaronson and Ambainis, building on a connection between Forrelation and a class of Boolean functions, the \emph{bent boolean functions} (see, e.g.,~\cite{dillon1974elementary}), characterized by having an ``extremely flat'' Fourier spectrum. While~\cite{GirishS26} pushed their construction (which used a subclass of the so-called Maiorana--McFarland bent functions) as far as possible, our result shows how to use a more general class of bent functions to obtain a stronger lower bound. Yet \emph{even our work only uses a tiny fraction of what is known (and still unknown) about bent functions}, and over the course of obtaining our result we discovered a wild, untapped trove of bent function families, with connections to various areas in Mathematics and Computer Science, and a host of open problems. We encourage the reader with a keen mathematical bent to see our work as an excuse to explore this landscape, and these connections.

\paragraph{Additional results.} As discussed earlier, similar problems have also been studied for the purposes of finding quantum speedups. One of the first, which served as inspiration for Shor's celebrated factoring algorithm, was Simon's algorithm for finding a hidden period \cite{DBLP:journals/siamcomp/Simon97a}. The generalization of this problem, known broadly as \emph{hidden subgroup finding}, has been the subject of much fruitful study. One such generalization, studied by Ye et al.~\cite{DBLP:journals/iandc/YeHLW21}, is the Generalized Simon's Problem (defined in \cref{sec:gsp_proof}), where the period is generalized to a group, instead of a period, and the field is over some prime, rather than Boolean values. They were able to demonstrate that the advantage achieved by Simon was generalizable to this extent.

We then show how some simple ideas (explored in the process of obtaining~\cref{thm:main}) yield a tight lower bound for the Generalized Simon's Problem (defined in \cref{sec:gsp_proof}). At a high level, here one is given a function $f\colon\F^n\rightarrow X$, which has a hidden subgroup $S$ of dimension $k$ embedded in it. The task is to identify this subgroup with as few queries as possible. We prove a tight lower bound for the query complexity:
\begin{theorem}
    \label{thm:generalized:simons}
    Any (possibly adaptive) classical algorithm for the Generalized Simon's Problem requires at least $\Omega(k,\sqrt{kp^{n-k}})$ queries. Moreover, this is optimal.
\end{theorem}
This resolves an open question of Ye~et al.~\cite{DBLP:journals/iandc/YeHLW21}, who established the upper bound but only showed an $\Omega(k,\sqrt{p^{n-k}})$ query complexity lower bound.

\paragraph{Notation.} Throughout, we assume without (much) loss of generality that $n$ is a multiple of $4$, and let $m=n/2$ and $M=2^{m}$. We will work in
$\F_2^n$ with the inner product
$\langle x,y\rangle=\sum_{k=1}^{n}x_ky_k$ for $x,y\in \F_2^n$, computed in $\F_2$: in particular, all sums are for the addition in $\F_2$, and we write $x+y$ for the coordinate-wise addition. For
$f\colon\F_2^n\to\F_2$ we denote by $\chi_f$ the
$\{-1,1\}$-valued function $(-1)^{f}$. The Fourier analysis and linear algebra notions we will use are recalled in~\cref{sec:preliminaries}.

\subsection{High-level overview}
\paragraph{Refresher: lower bounds, Yao's Principle, and transcripts.} We will be concerned with lower bounds for \emph{non-adaptive} algorithms, that is, algorithms that fix all the points they will query before receiving any query answer. Any such $q$-algorithm is defined by a set
$Q\subseteq\F_2^n$ of $q$ points together with a decision rule mapping the answers to a bit (its decision). At every queried point $x\in Q$, we can assume the algorithm receives
the value of both functions,\footnote{Answering with both functions only helps the algorithm, at a cost of a factor of
2 in the query complexity, so for the sake of our argument we can assume without loss of generality that on query $x$ the algorithm gets both $f(x)$ and $g(x)$.} so the answers form the \emph{transcript}
\begin{equation*}
    \tau=\big(f(x),g(x)\big)_{x\in Q},
\end{equation*}
a list of $2q$ values.

To prove a lower bound~--~that, for every randomized algorithm, there exists a worst-case instance on which it will fail, by Yao's principle it is enough to exhibit a single distribution of \yes- and \no-instances which is hard for all \emph{deterministic} algorithms. Since every non-adaptive deterministic algorithm then corresponds to a fixed set of queries, it will be enough to show that, for every set $Q$, the two distributions over transcripts induced by the answers in the \yes- and \no-cases are indistinguishable unless $|Q|$ is very large.

Our distribution of instances, described in the next subsection, will depend on a ``secret bit'' $s$ that determines whether the instance is a
\yes or \no. Let $\mu_0$ and $\mu_1$ be the distributions of $\tau$
when the instance is drawn with $s=0$ and with $s=1$. The \emph{total
variation distance} between two such distributions is

\begin{equation*}
    \|\mu_0-\mu_1\|_{\tv}=\tfrac12\sum_{\tau \text{ transcript}}|\mu_0(\tau)-\mu_1(\tau)|\,.
\end{equation*}
If an
algorithm outputs the correct value of $s$ with probability at least $2/3$,
then $\|\mu_0-\mu_1\|_{\tv}\geq1/3$. Thus an upper bound on this distance that
holds for every query set of size $q$ yields a lower bound on the number of
queries a non-adaptive algorithm needs to solve extremal Forrelation.

\subsubsection{Our Hard Instances Construction}\label{sec:construction}
Our hard instances $\chi_f,\chi_g\colon \F^n_2 \to \{\pm1\}$ are built as follows. We first define
$f,g\colon\F_2^n\to\F_2$ from the cosets of a random subspace of $\F^n_2$; and then set $(\chi_f,\chi_g)=((-1)^f, (-1)^g)$. In more detail, we choose:
\begin{enumerate}
    \item A uniformly random $m$-dimensional subspace $K$ of
          $\F_2^n$, so that $K^\perp$ has dimension $m$ as well. Let
          $C_1,\dots,C_M$ be the cosets of $K$ and $D_1,\dots,D_M$ those of
          $K^\perp$, enumerated in any fixed way, and for $i,j\in[M]$ let
          $\phi_i$ and $\psi_j$ be the smallest elements of $C_i$ and $D_j$,
          reading vectors as binary numbers.\footnote{This ordering is arbitrary, and any (arbitrary, but fixed) way to choose a coset representative would be suitable.}
          These are the \emph{representatives} of their cosets.
    \item A uniformly random permutation $\pi$ of $[M]$, which pairs the
          coset $C_i$ with the coset $D_{\pi(i)}$.
    \item A uniformly random function $h\colon[M]\to\F_2$. We call
          $h(i)$ the \emph{random bit} of the coset $C_i$.
\end{enumerate}
Fix also a bit $s\in\F_2$, which will decide the type of instance we
build: as we will see, $s=0$ gives a \yes-instance, and $s=1$, a \no-instance. We further assign to each coset
of $K^\perp$ its ``own random bit'' by setting, for $j\in[M]$,
\begin{equation}\label{eq:tie}
    \bar{h}(j)=h(i)+\langle\phi_i,\psi_j\rangle+s,
    \qquad\text{where }i=\pi^{-1}(j).
\end{equation}
Notice that since $h$ is uniform, so is $\bar{h}$; however, $h$ and $\bar{h}$ are \emph{not} independent.

We are ready to define our two functions $f,g$ in a unified way, setting, for all $i,j\in [M]$,
\begin{align}
    f(x)&=\langle x,\psi_{\pi(i)}\rangle+h(i)&&\text{for }x\in C_i,
        \label{eq:def_f}\\
    g(y)&=\langle\phi_{\pi^{-1}(j)},y\rangle+\bar{h}(j)&&\text{for }y\in D_j.
        \label{eq:def_g}
\end{align}
Each function is affine within the cosets of $K$ and $K^\perp$ respectively,
with the random bit of the coset as the constant term.
In \Cref{fig:construction} we show an example of the construction for $n=4$.

\begin{figure}[t]
    \centering
    \begin{tikzpicture}[>=stealth,
            pt/.style={inner sep=1.5pt,font=\ttfamily},
            rep/.style={pt,fill=black!15,rounded corners=1pt}]
        \def\boxpad{0.1}
        \def\boxgap{0.3}
        \pgfmathsetmacro{\boxhalf}{0.28+\boxpad}
        \pgfmathsetmacro{\rowpitch}{2*\boxhalf+\boxgap}

        \node at (2,{-\rowpitch+\boxhalf+0.37}) {cosets of $K$};
        \node at (8.5,{-\rowpitch+\boxhalf+0.37}) {cosets of $K^\perp$};
        \node at (5.25,{-\rowpitch+\boxhalf+0.37}) {$\pi$};

        \foreach \row/\pa/\pb/\pc/\pd in {1/0000/0001/1110/1111,
                2/0010/0011/1100/1101, 3/0100/0101/1010/1011,
                4/0110/0111/1000/1001} {
            \draw[rounded corners] ({-\boxpad},{-\rowpitch*\row-\boxhalf})
                rectangle ({4+\boxpad},{-\rowpitch*\row+\boxhalf});
            \node[anchor=east] at ({-\boxpad-0.2},{-\rowpitch*\row})
                {$C_{\row}$};
            \node[rep] at (0.5,{-\rowpitch*\row}) {\pa};
            \node[pt] at (1.5,{-\rowpitch*\row}) {\pb};
            \node[pt] at (2.5,{-\rowpitch*\row}) {\pc};
            \node[pt] at (3.5,{-\rowpitch*\row}) {\pd};
        }
        \foreach \row/\pa/\pb/\pc/\pd in {1/0000/0110/1010/1100,
                2/0001/0111/1011/1101, 3/0010/0100/1000/1110,
                4/0011/0101/1001/1111} {
            \draw[rounded corners] ({6.5-\boxpad},{-\rowpitch*\row-\boxhalf})
                rectangle ({10.5+\boxpad},{-\rowpitch*\row+\boxhalf});
            \node[anchor=west] at ({10.7+\boxpad},{-\rowpitch*\row})
                {$D_{\row}$};
            \node[rep] at (7,{-\rowpitch*\row}) {\pa};
            \node[pt] at (8,{-\rowpitch*\row}) {\pb};
            \node[pt] at (9,{-\rowpitch*\row}) {\pc};
            \node[pt] at (10,{-\rowpitch*\row}) {\pd};
        }

        \foreach \i/\j in {1/2, 2/4, 3/1, 4/3}
            \draw[->] ({4.1+\boxpad},{-\rowpitch*\i})
                -- ({6.4-\boxpad},{-\rowpitch*\j});
    \end{tikzpicture}
    \caption{An instance of the construction for $n=4$ with $K=\{\texttt{0000},\texttt{0001},\texttt{1110},\texttt{1111}\}$.
    Each box is a coset, so every point of $\F_2^4$ appears once on each side. The shaded points are the representatives. The middle arrows denote the permutation $\pi$. Here $f(x)=\langle x,\texttt{0001}\rangle+h(1)$ for $x\in C_1$, while $g(y)=\langle\texttt{0100},y\rangle+\bar{h}(1)$ for 
    $y\in D_1$. Notice that the coset $C_3$ is paired with $D_1$ and
    $\langle\phi_3,\psi_1\rangle=0$, so~\eqref{eq:tie} in this case means
    $h(3)=\bar{h}(1)+s$.}%
    \label{fig:construction}
\end{figure}

\paragraph{What is new here, and why so much randomness?} It is natural to wonder how our instances relate to the previous construction by Girish and Servedio, and, at a more abstract level, why we bother adding \emph{three} sources of randomness in our constructions (the random subspace $K$, the random permutation $\pi$, and the random function $h$): could we have done as well without all of this? Peeking ahead,
\Cref{sec:gs_relation} relates this construction to the Maiorana--McFarland
class of bent functions and to the construction of Girish and Servedio, and
\cref{sec:upper_bounds} asks how many queries suffice when one of
the three random choices above (namely, either the permutation $\pi$, or the random subspace $K$)  is fixed and known to the algorithm.

\subsubsection{Technical overview of the proof}
\label{sec:overview}

Fix a non-adaptive algorithm and let $Q$ be its set of $q$ queries. We show that
when $q$ is below $2^{(1/2-o(1))n}$ the transcripts of \yes and \no
instances are statistically indistinguishable, as $\|\mu_0-\mu_1\|_{\tv} = o(1)$. To do so, rather than comparing $\mu_0$ with
$\mu_1$ directly, we show that each of them is close to a single \emph{reference distribution}~$\rho$
that does not depend on $s$.

Let $\rho$ be the distribution of the
transcript when $K$, $\pi$ and $h$ are sampled as in \cref{sec:construction}
but $\bar{h}$ is sampled uniformly and independently of $h$ rather than set
by~\eqref{eq:tie}, with $f$ and $g$ still given by~\eqref{eq:def_f}
and~\eqref{eq:def_g}. %

Under both $\mu_0$ and $\mu_1$, each answer in the transcript, seen in
isolation, is a uniform bit, and only the dependence between the answers in cosets
matched by $\pi$ carries any information about $s$. Since that dependence is
absent under $\rho$, $\mu_s$ will be close to $\rho$ whenever a transcript reveals little about
which queried cosets $\pi$ pairs.

Our argument that $\mu_s$ is close to $\rho$ works when, in every coset of the
subspace $K$ and of $K^\perp$, the differences of the queries from $Q$ span a
subspace of small dimension (at most some parameter $r$): we say that a subspace $K$ satisfying this (for a fixed set of queries $Q$) is ``$r$-good.'' For some intuition, we encourage the reader to consider the case where each coset of $K$ and $K^\perp$ contains at most $2$ queries from $Q$, i.e., when \emph{there is no 3-collision in any coset}, which would correspond to $r=1$: the notion of ``$r$-good'' is a suitable generalization of this absence of collisions. We then use that the total variation (TV) distance between $\mu_0$ and $\mu_1$ satisfies
\begin{align}
    \label{eq:intro:step1}
        \|\mu_0-\mu_1\|_{\tv}
        &\leq\Pr[K\text{ is not $r$-good}]
            +2\cdot \smash{\max_{K\text{ $r$-good}}}
             \|\mu_1(\cdot\mid K)-\rho(\cdot\mid K)\|_{\tv} 
\end{align}
and any successful algorithm must have $\|\mu_0-\mu_1\|_{\tv} = \Omega(1)$. Bounding the first term of the RHS is relatively straightforward: we show in \cref{sec:good_subspace} (\cref{lem:good_subspace}) that
\begin{equation}
    \label{eq:intro:step2}
\Pr[K\text{ is not $r$-good}] \leq q^{r+2}M^{-(r+1)}\,.
\end{equation}
The crux of the argument is in controlling the second term, before balancing both bounds (by setting $r$ carefully) to obtain our main result.\smallskip

\noindent\textbf{It's all about matchings!} To do so, for a fixed $r$-good subspace $K$, we introduce the notion of the \emph{pinning graph} $H_\tau$ of a transcript $\tau$ (\cref{def:pinning_graph}), a bipartite graph between cosets of $K$ and $K^\perp$ whose edges encode the constraints on the permutation $\pi$ revealed by the queries: this pinning graph exactly captures the amount of information revealed about $\pi$ to the algorithm, in that its perfect matchings are exactly the permutations consistent with the transcript $\tau$. This viewpoint enables us to reduce the analysis to studying the number of matchings of $H_\tau$, for a transcript $\tau$ consistent with a permutation~$\pi$: specifically,~\cref{lem:rho_given_pi,lem:mu_given_pi} first lets us exactly compute the probability under $\mu_s$ and $\rho$ (conditioned on $K$) of any transcript $\tau$, as a function of the number of queries in cosets of $K$ and $K^\perp$ and how many of them ``hit'' a pair matched by $\pi$. Using this, \cref{lem:reduction} then bounds the TV distance
$
        \|\mu_1(\cdot\mid K)-\rho(\cdot\mid K)\|_{\tv}
$
by a simple function of the \emph{matching probabilities} $p_{I,J}(H_\tau)$ of the pinning graph $H_\tau$ of a random transcript $\tau\sim \rho(\cdot\mid K)$, which are the probabilities (for subsets $I,J$ of queried cosets) that a random matching $\pi$ of $H_\tau$ maps $I$ to $J$:
    \begin{equation*}
        \|\mu_1(\cdot\mid K)-\rho(\cdot\mid K)\|_{\tv}
        \leq \frac{1}{2} \Big(\mathbb{E}_{\tau\sim\rho(\cdot\mid K)}\Big[
            \sum_{I,J\neq \emptyset}
            p_{I,J}(H_\tau)^{2}\Big]\Big)^{1/2}.
    \end{equation*}
The last (and final) ingredient is then to convert this upper bound into a function of $q,M$, and $r$ only. This is done in~\cref{lem:switching}, which bounds the matching probabilities $p_{I,J}(H_\tau)$ by the probability
that the sub-matching~$E(\pi)$ of $\pi$ revealed by the $q$ queries contains a given matching $T$, for a
uniformly random perfect matching $\pi$ of $H$. The key is to show that such a random $\pi$ contains $T$ with
probability at most $(4^{r+1}/M)^{|T|}$: this follows by a ``switching argument'' (\cref{lem:switching:claim2}), which bounds the number of matchings $\pi$ which extend any given $T$. Roughly speaking, we show that starting from a matching $T'$ compatible with the transcript $\tau$ and extending it by an edge $e$, the probability to get a new matching $T=T'\cup\{e\}$ compatible with $\tau$ is small. By applying this edge by edge, starting with the empty matching to build our target $T$, this lets us bound the probability that $\pi$ contains $T$ overall.  From~\cref{lem:switching}, it is then straightforward to show that
\begin{equation}
    \label{eq:intro:step3}
        \|\mu_1(\cdot\mid K)-\rho(\cdot\mid K)\|_{\tv}
        \leq O\!\left(\frac{4^{r}q}{M}\right)
\end{equation}
which is done in~\cref{lem:pin_sum}. Combining~\eqref{eq:intro:step2} and~\eqref{eq:intro:step3} into the RHS of~\eqref{eq:intro:step1} and setting $r=\Theta(\sqrt{m})$ then yields
\[
    \|\mu_0-\mu_1\|_{\tv} = o(1)
\]
unless $q > 2^{n/2-O(\sqrt{n})}$, establishing the lower bound.

\paragraph{Organization of the paper.} After recalling useful notions and facts in~\cref{sec:preliminaries}, we prove our main result in~\cref{sec:proof}.~\cref{sec:gs_relation} provides more context on the class of bent functions our construction relies on, the Maiorana--McFarland Class, and how ours relates to the construction of Girish and Servedio~\cite{GirishS26};~\cref{sec:upper_bounds} then gives a matching (non-adaptive) upper bound for our construction, and justifies the choices made in its definition, explaining how removing any of its randomly chosen components (the subspace $K$ or the permutation $\pi$) would lead to an algorithm ``beating it'' with few queries; before explaining why non-adaptivity is necessary for our construction. Finally,~\cref{sec:gsp_proof} provides the proof of our ``side result,''~\cref{thm:generalized:simons}.

\begin{framed}
\noindent\textbf{AI methodologies.} The high-level idea of the construction and the gist of the argument behind~\cref{thm:main:forrelation} were formulated by the authors. To make this argument rigorous, the authors then used Anthropic Claude Opus 5, over interactions spanning July to September 2026. The final proof was written, simplified, and verified by the authors.
\end{framed}
\section{Preliminaries}\label{sec:preliminaries}
In the next two subsections, we formally define the notions we will rely on, and record a few useful facts which are used in our proofs. 
\subsection{Boolean Functions}
\begin{definition}[Fourier Transform]\label{def:fourier_forrelation}
    For $f\colon \F_2^n\to\mathbb{R}$, the \emph{Fourier transform}
    of $f$ is the function $\widehat{f}\colon \F_2^n\to\mathbb{R}$ defined by
    \begin{equation*}
        \widehat{f}(y)
            =2^{-n}\sum_{x\in\F_2^n}f(x)(-1)^{\langle x,y\rangle},
            \qquad y\in\F_2^n.
    \end{equation*}
\end{definition}
\noindent Given this, we can next define the Forrelation function:
\begin{definition}[Forrelation of two functions]\label{def:forrelation_function}
    For Boolean functions $f,g\colon\F_2^n\to\{-1,1\}$, their
    \emph{forrelation} is
    \begin{equation*}
        \forr(f,g)=2^{-n/2}\sum_{x\in\F_2^n}\widehat{f}(x)g(x).
    \end{equation*}
\end{definition}
Note that the Forrelation takes values in $[-1,1]$: indeed, by Cauchy--Schwarz,
\begin{equation*}
    |\forr(f,g)|\leq2^{-n/2}\sqrt{\sum_{x\in\F_2^n}\widehat{f}(x)^2}
    \,\sqrt{\sum_{x\in\F_2^n}g(x)^2},
\end{equation*}
where the first sum equals $2^{-n}\sum_{x}f(x)^2=1$ by Parseval's
identity and the second equals $2^n$. One can thus think of $\forr(f,g)$ as taking the correlation of $g$ with the Fourier transform of $f$, hence the name. We will also rely on a special class of boolean functions, called \emph{bent} functions, as discussed above.
\begin{definition}[Bent Boolean Functions]\label{def:bent_boolean_function}
    A Boolean function $f$ on $n$ variables is said to be \emph{bent} if $|\hat{f}(y)|=2^{-n/2}$ for all $y\in\F_2^n$.
\end{definition}
\noindent We can now define the main problem we study, (extremal) Forrelation:
    \begin{definition}[(Extremal) Forrelation Problem]\label{def:extremal_Forrelation_problem}
        Let $0<\varepsilon\leq 1$ be a fixed constant. Then, an adversary is given query access to the functions $f$ and $g$, with the promise that either:
        \begin{enumerate}
            \item $\forr(f,g)\geq \varepsilon$, known as a \yes-instance, or
            \item $\forr(f,g)\leq-\varepsilon$, known as a \no-instance.
        \end{enumerate}
        The \emph{Forrelation problem} asks what the minimum number of queries needed is for an adversary to succeed at distinguishing the \yes-instance from the \no-instance with high (constant) probability. When $\varepsilon=1$, we refer to it as the \emph{extremal} Forrelation problem, or simply extremal Forrelation.
    \end{definition}

Our hard instances for extremal Forrelation are described in \cref{sec:construction}: they will be bent functions, and will be constructed from the cosets of
a subspace of $\F_2^n$, which we define in the following section.
\subsection{Linear Algebraic Definitions}
We start by defining a coset.
\begin{definition}[Coset]\label{def:coset}
    Let $V$ be a subspace of $\F_2^n$ and let $v\in\F_2^n$. The
    \emph{coset} of $V$ containing $v$ is $v+V=\{v+u:u\in V\}$.
\end{definition}

Note that the cosets of $V$ partition $\F_2^n$ into $2^n/|V|$ sets, with two vectors $x,x'$ belonging to the same coset if, and only if, their sum $x+x'$ is in $V$. Write $V^\perp$
for the orthogonal complement of $V$. The next lemma states three properties of
the cosets of a subspace and of $V^\perp$ that will be useful in the analysis. These results are standard, and we include proofs for completeness.

\begin{lemma}\label{lem:coset_pairing}
    Let $K$ be a subspace of $\F_2^n$ with $\dim(K)=m$.
    \begin{enumerate}
        \item For $u\in K$ and a coset $D$ of $K^\perp$, the value
              $\langle u,y\rangle$, denoted $\langle u,D\rangle$, is the same for every $y\in D$.
        \item For distinct cosets $D$ and $D'$ of $K^\perp$ and a basis
              $u_1,\dots,u_m$ of $K$,
              $\langle u_\ell,D\rangle\neq\langle u_\ell,D'\rangle$ for some
              $1\leq \ell\leq m$.
        \item For linearly independent $u_1,\dots,u_t\in K$ and bits
              $b_1,\dots,b_t$, exactly $M/2^{t}$ of the $M$ cosets of
              $K^\perp$ satisfy $\langle u_\ell,D\rangle=b_\ell$ for all
              $1\leq\ell\leq t$.
    \end{enumerate}
    The same holds with $K$ and $K^\perp$ interchanged.
\end{lemma}
\begin{proof}
    For the first claim, if $y,y'\in D$ then $y+y'\in K^\perp$, so
    $\langle u,y\rangle+\langle u,y'\rangle=\langle u,y+y'\rangle=0$ and the two
    values agree. For the second, take $y\in D$ and $y'\in D'$. As the two
    cosets are distinct, $y+y'\notin K^\perp$, so $\langle u,y+y'\rangle=1$ for
    some $u\in K$. Therefore $\langle u_\ell,y+y'\rangle=1$ for some
    $1\leq \ell\leq m$, which in turn implies
    $\langle u_\ell,D\rangle = \langle u_\ell,y\rangle
    \neq\langle u_\ell,y'\rangle = \langle u_\ell,D'\rangle$.

    For the third, we extend $u_1,\dots,u_t$ to a basis $u_1,\dots,u_m$ of $K$ and
    map each coset $D$ of $K^\perp$ to
    $(\langle u_1,D\rangle,\dots,\langle u_m,D\rangle)$, which the first claim
    makes well defined. By the second claim this map is injective, and it goes
    from a set of size $M$ to one of size $2^m=M$, so it is a bijection.
    Exactly $2^{m-t}=M/2^{t}$ points of $\F_2^m$ have first $t$ coordinates
    $b_1,\dots,b_t$.

    \noindent Swapping the roles of $K$ and $K^\perp$ yields the last sentence of the
    lemma,  using that $(K^\perp)^\perp=K$.
\end{proof}

Several of our Fourier computations will rely on computing the sum of $(-1)^{\langle x,v\rangle}$ over a
coset: 
\begin{lemma}\label{lem:character_sum}
    Let $V$ be a subspace of $\F_2^n$, and let $C$ be a coset of $V$,
    $c\in C$, and $v\in\F_2^n$. Then
    \begin{equation*}
        \sum_{x\in C}(-1)^{\langle x,v\rangle}=\begin{cases}
            (-1)^{\langle c,v\rangle}|V|,&\text{if }v\in V^\perp,\\
            0,&\text{otherwise.}
        \end{cases}
    \end{equation*}
\end{lemma}
\begin{proof}
    Every $x\in C$ is $c+u$ for exactly one $u\in V$, namely $u=x+c$, so
    \begin{equation*}
        \sum_{x\in C}(-1)^{\langle x,v\rangle}
        =(-1)^{\langle c,v\rangle}\sum_{u\in V}(-1)^{\langle u,v\rangle}.
    \end{equation*}
    If $v\in V^\perp$ then every term of the remaining sum is $1$ and the sum is
    $|V|$. Otherwise the definition of $V^\perp$ implies the existence of some $u'\in V$ with
    $\langle u',v\rangle=1$, and the pairs $\{u,u+u'\}$ partition $V$; the two
    terms of such a pair are $(-1)^{\langle u,v\rangle}$ and
    $(-1)^{\langle u,v\rangle}(-1)^{\langle u',v\rangle}$, which cancel.
\end{proof}

The instances of \cref{sec:construction} are built from a uniformly sampled
$m$-dimensional subspace, and three facts about that draw are used throughout.

\begin{lemma}\label{lem:subspace_count}
    Let $V$ be a uniformly random $m$-dimensional subspace of
    $\F_2^n$.
    \begin{enumerate}
        \item There are fewer than $2^{n^{2}/4+2}$ such subspaces.
        \item For any fixed $t\leq m$ linearly independent $v_1,\dots,v_t\in\F_2^n$, $\Pr[v_1,\dots,v_t\in V]\leq M^{-t}$.
        \item $V^\perp$ is uniformly distributed among the $m$-dimensional
              subspaces as well.
    \end{enumerate}
\end{lemma}
\begin{proof}
    A $d$-dimensional subspace contains $T(d,t) := \prod_{r=0}^{t-1}(2^{d}-2^{r})$ ordered
    linearly independent $t$-tuples for $t\leq d$, since after $r$ of them are
    chosen the next may be any point outside their span, which contains $2^{r}$
    points.

    For the first claim, every ordered linearly independent $m$-tuple of
    $\F_2^n$ is a basis of exactly one $m$-dimensional subspace, and
    each such subspace has $T(m,m) = \prod_{r=0}^{m-1}(2^{m}-2^{r})$ of them, so the
    number of subspaces is
    \begin{equation*}
        \frac{T(n,m)}{T(m,m)} = \prod_{r=0}^{m-1}\frac{2^{n}-2^{r}}{2^{m}-2^{r}}
        =2^{m(n-m)}\prod_{r=0}^{m-1}\frac{1-2^{r-n}}{1-2^{r-m}}
        <\frac{2^{n^{2}/4}}{\prod_{j\geq1}(1-2^{-j})},
    \end{equation*}
    using that $m(n-m)=n^{2}/4$. Since
    $\prod_{j\geq3}(1-2^{-j})\geq1-\sum_{j\geq3}2^{-j}=3/4$, the denominator is
    at least $\tfrac12\cdot\tfrac34\cdot\tfrac34>1/4$, giving an upper bound of $2^{n^{2}/4+2}$ on the number of subspaces.

    For the second claim, by symmetry the quantity $\Pr[v_1,\dots,v_t\in V]$ (where the probability is taken over the uniformly random choice of $V$) does not depend on the specific $t$-tuple of linearly independent vectors.
    Averaging over all of them makes it the ratio of the two counts,
    \begin{equation*}
        \Pr[v_1,\dots,v_t\in V]
        = \mathbb{E}_{(v'_1,\dots,v'_t)}[\Pr[v'_1,\dots,v'_t\in V]]
        = \frac{T(m,t)}{T(n,t)}
        = \prod_{r=0}^{t-1}\frac{2^{m}-2^{r}}{2^{n}-2^{r}}\leq 2^{t(m-n)}=M^{-t},
    \end{equation*}
    each factor being at most $2^{m-n}$ because $m\leq n$, using $\frac{a-c}{b-c}\leq \frac{a}{b}$ for $c < a < b$.

    For the third claim, taking orthogonal complements sends an $m$-dimensional
    subspace to one of dimension $n-m=m$, and it is injective because a
    subspace is the orthogonal complement of its own complement, so it is a
    bijection of the $m$-dimensional subspaces onto themselves and maps the
    uniform distribution to itself.
\end{proof}

\section{Main result: a \texorpdfstring{$2^{(\frac12-o(1))n}$}{2 to the (1/2 - o(1))n} Non-Adaptive Lower Bound}\label{sec:proof}

\subsection{The Result}
In this section, we establish our main result,~\cref{thm:main:forrelation}, restated below:
\begin{theorem}\label{thm:main}
    For every non-adaptive algorithm making $q$ queries and every integer
    $1\leq r<m$,
    \begin{equation}\label{eq:main_bound}
        \|\mu_0-\mu_1\|_{\tv}
            \leq 2\binom{q}{r+2}M^{-(r+1)}+2\cdot4^{r+1}\frac{q}{M}.
    \end{equation}
    In particular, every non-adaptive algorithm that solves extremal Forrelation with
    probability at least $2/3$ must make $2^{(1/2-o(1))n}$ queries, where the $o(1)$ term is
$O(1/\sqrt{n})$.
\end{theorem}
To see where the two terms of the RHS come from, recall from~\cref{sec:construction} that our instances are built from a random subspace $K$
and a random permutation $\pi$ that pairs the cosets of $K$ with those of
$K^\perp$: the construction is such that an algorithm that knows $K$ and $\pi$
can solve the problem with two queries. As discussed in~\cref{sec:overview}, we show
that when the queries made by the algorithm have few collisions in these
cosets, a transcript reveals little about $\pi$
unless $q$ is large. Specifically, we require that no coset of $K$ or $K^\perp$
hold queries whose affine span has dimension more than $r$ (i.e., that ``$K$ is $r$-good''). The first term
then bounds the probability that $K$ is \emph{not} $r$-good, and the second bounds the distance
between $\mu_0$ and $\mu_1$ when it is. For $r=1$ the first
term dominates and the resulting lower bound on $q$ is $\Omega(2^{n/3})$, and for
$r=2$ it becomes $\Omega(2^{3n/8})$. Letting $r\approx\sqrt{n}$
balances the two terms and yields our $2^{(1/2-o(1))n}$ lower bound. We show in~\cref{sec:subspace_algorithm} how to distinguish our \yes-
and \no-instances non-adaptively in $O(2^{n/2})$ queries, so the bound is
optimal for these instances up to the $o(1)$ term in the exponent.

\paragraph{The plan.} We begin by recording some important properties of our construction in~\cref{ssec:properties} (for instance, that they are, indeed, instances of extremal Forrelation). \cref{sec:good_subspace} then shows that a random $K$ is ``$r$-good'' with high probability.
\cref{sec:transcripts} derives explicit formulas for $\mu_s$ and $\rho$, and
\cref{sec:reduction} bounds the TV distance between them by how well a transcript
determines which queried cosets $\pi$ pairs. Assuming $K$ is $r$-good,
\cref{sec:switching} shows that a transcript determines this poorly.
Finally, \Cref{sec:theorem_proof} puts everything  together to obtain the lower bound.

\subsection{Some properties of our hard instances}
    \label{ssec:properties}
As pointed out earlier, $h$ and $\bar{h}$ defined in the construction of our functions $f,g$ are not independent: as the next lemma shows, this dependence can be leveraged to compute the secret bit 
$s$ using two queries in paired cosets.
\begin{lemma}[Queries in matched cosets give away the secret bit]\label{lem:coset_rep_independence}
    For every $i\in[M]$, $x\in C_i$ and $y\in D_{\pi(i)}$, we have
    $g(y)=f(x)+\langle x,y\rangle+s$.
\end{lemma}
\begin{proof}
    Fix $i\in[M]$, $x\in C_i$ and $y\in D_{\pi(i)}$.
    Adding~\eqref{eq:def_f} to~\eqref{eq:def_g} and
    substituting~\eqref{eq:tie}, the two copies of $h(i)$ cancel and
    \begin{equation*}
        f(x)+g(y)=\langle x,\psi_{\pi(i)}\rangle
                  +\langle\phi_i,y\rangle
                  +\langle\phi_i,\psi_{\pi(i)}\rangle+s.
    \end{equation*}
    Since $x$ and $\phi_i$ are in the same coset, and similarly for $y$ and $\psi_{\pi(i)}$, we can write $x=\phi_i+u$ with $u\in K$ and $y=\psi_{\pi(i)}+v$ with
    $v\in K^\perp$. Then
    \begin{align*}
        \langle x,\psi_{\pi(i)}\rangle
            &=\langle\phi_i,\psi_{\pi(i)}\rangle
             +\langle u,\psi_{\pi(i)}\rangle,\\
        \langle\phi_i,y\rangle
            &=\langle\phi_i,\psi_{\pi(i)}\rangle
             +\langle\phi_i,v\rangle,
    \end{align*}
    so plugging this in the above expression of $f(x)+g(y)$ and cancelling things out, we get
    \begin{equation*}
        f(x)+g(y)=\langle\phi_i,\psi_{\pi(i)}\rangle
                  +\langle u,\psi_{\pi(i)}\rangle
                  +\langle\phi_i,v\rangle+s.
    \end{equation*}
    On the other hand $\langle u,v\rangle=0$, because $u\in K$ and
    $v\in K^\perp$, so
    \begin{align*}
        \langle x,y\rangle + s
        &=\langle\phi_i+u,\psi_{\pi(i)}+v\rangle  + s\\
        &=\langle\phi_i,\psi_{\pi(i)}\rangle+\langle\phi_i,v\rangle
         +\langle u,\psi_{\pi(i)}\rangle  + s\,.
    \end{align*}
   This shows that $f(x)+g(y)=\langle x,y\rangle+s$, concluding the proof.
\end{proof}

It remains to check that the Forrelation of the pair is $(-1)^s$, which yields
the promise $|\forr(\chi_f,\chi_g)|=1$ and makes $s$ the answer bit. This will
follow from an identity characterizing the Fourier coefficients of $\chi_f$ and $\chi_g$:
\begin{lemma}\label{lem:fourier_identity}
    For every $y\in\F_2^n$,
    $\widehat{\chi_f}(y)=2^{-n/2}(-1)^s\chi_g(y)$.
\end{lemma}
\begin{proof}
    Fix $y$, let $j\in[M]$ be the index with $y\in D_j$, and let
    $i=\pi^{-1}(j)$. The cosets of $K$ partition $\F_2^n$, and on
    $C_{i'}$ we have $f(x)=\langle x,\psi_{\pi(i')}\rangle+h(i')$, so
    splitting the Fourier transform over the cosets gives
    \begin{align*}
        \widehat{\chi_f}(y)&=2^{-n}\sum_{x\in\F_2^n}
          (-1)^{f(x)}(-1)^{\langle x,y\rangle}\\
        &=2^{-n}\sum_{i'\in[M]}(-1)^{h(i')}
          \sum_{x\in C_{i'}}(-1)^{\langle x,y+\psi_{\pi(i')}\rangle}.
    \end{align*}
    Fix any $i'$ and write $v=y+\psi_{\pi(i')}$. By \cref{lem:character_sum}
    at the coset $C_{i'}$ of $K$ with representative $\phi_{i'}$, the inner
    sum is $(-1)^{\langle\phi_{i'},v\rangle}M$ when $v\in K^\perp$ and $0$
    otherwise. Since $\psi_{\pi(i')} \in D_{\pi(i')}$, the condition
    $v\in K^\perp$ holds if and only if $y \in D_{\pi(i')}$. Only $i'=i$ does
    that, and there $v=y+\psi_j$, so
    \begin{align*}
        \widehat{\chi_f}(y)&=2^{-n}(-1)^{h(i)}
          (-1)^{\langle\phi_i,y+\psi_j\rangle}M\\
        &=2^{-n/2}(-1)^{\langle\phi_i,y\rangle
                        +\langle\phi_i,\psi_j\rangle+h(i)+s}(-1)^s\\
        &=2^{-n/2}\chi_g(y)(-1)^s,
    \end{align*}
    where the last line is~\eqref{eq:def_g} at $y$, with $\pi^{-1}(j)=i$ and
    $\bar{h}(j)$ expanded by~\eqref{eq:tie}.
\end{proof}
This implies that $\chi_f,\chi_g$ are indeed, depending on the bit $s$, \yes- or \no-instances of extremal Forrelation. As a (nice) aside, this also immediately shows that $\chi_f$ and  $\chi_g$ are bent.
\begin{corollary}\label{lem:forr_value}
Both  $\chi_f$ and  $\chi_g$ are bent functions. Moreover, $\forr(\chi_f,\chi_g)=(-1)^s$.
\end{corollary}
\begin{proof}
Since $\chi_g$ takes values in $\{-1,1\}$, \cref{lem:fourier_identity} gives
$|\widehat{\chi_f}(y)|=2^{-n/2}$ for every $y$, so $\chi_f$ is bent. The corresponding claim for $\chi_g$ follows by symmetry in their definition.

Turning to the value of their Forrelation, observe that, from the definition of Forrelation, we have
    \begin{align*}
        \forr(\chi_f,\chi_g)
        &=2^{-n/2}\sum_{x\in\F_2^n}\widehat{\chi_f}(x)\chi_g(x)\\
        &=2^{-n/2}\sum_{x\in\F_2^n}2^{-n/2}\chi_g(x)^2(-1)^s \tag{by \cref{lem:fourier_identity}}\\
        &=(-1)^s \tag{since $\chi_g(x)^2=1$ for all $x$},
    \end{align*}
    establishing the claim.
\end{proof}

\subsection{Good Subspaces, and Where to Find Them}
\label{sec:good_subspace}

We now give the formal definition of the condition on $K$ needed for our analysis. We then argue that this condition holds with overwhelming probability for a uniformly random choice of subspace $K$.

\begin{definition}[Goodness]\label{def:good}
    Given a set of queries $Q$ and an integer $r\geq1$, a subspace $K$ is
    \emph{$r$-good for $Q$} if, for every coset $C$ of $K$ or
    $K^\perp$, the differences $x+x'$ of the queries $x,x'\in Q\cap C$
    span a subspace of dimension at most $r$:
    that is,
    \[
    \dim\operatorname{span} \{x+x' : x,x' \in Q\cap C\} \leq r\,.
    \]
    This is known as the \emph{affine span} of $Q\cap C$.
\end{definition}

When $Q$ and $r$ are clear from context we simply say that $K$ is
\emph{good}. For $r=1$ the condition says that no coset holds three or more
queries, since the differences of 3 distinct queries span a subspace of
dimension $2$. Our next lemma shows that a uniformly random $K$ is good with high probability:
\begin{lemma}\label{lem:good_subspace}
    For every set $Q$ of $q$ queries and every integer $1\leq r<m$, $K$ is
    $r$-good with probability at least $1-2\binom{q}{r+2}M^{-(r+1)}$.
\end{lemma}
\begin{proof}
    Suppose the differences of the queries in some coset of $K$ span more than
    $r$ dimensions, and let $x_0$ be one of these queries. Every difference
    $x+x'$ is the sum of $x_0+x$ and $x_0+x'$, so some $r+1$ of the vectors
    $x_0+x$ are linearly independent, say $x_0+x_1,\dots,x_0+x_{r+1}$. So the
    coset contains $r+2$ queries whose differences span $r+1$ dimensions.

    Fix $r+2$ queries whose differences span $r+1$ dimensions. They share a
    coset of $K$ exactly when $K$ contains a basis of that span, and
    \cref{lem:subspace_count} with $t=r+1\leq m$ bounds the probability of this
    event by $M^{-(r+1)}$. By \cref{lem:subspace_count}, $K^\perp$ is also
    uniformly distributed among the $m$-dimensional subspaces, so the same
    bound holds for the event that the $r+2$ queries share a coset of
    $K^\perp$. Each of the at most $\binom{q}{r+2}$ such sets of queries thus
    gives two events to avoid, and $K$ is good when none of these
    $2\binom{q}{r+2}$ events occurs: by a union bound,
    $\Pr[K\text{ is not good}]\leq 2\binom{q}{r+2}M^{-(r+1)}$.
\end{proof}

\subsection{Transcripts, Pinning Graphs, and Compatible Permutations}
\label{sec:transcripts}

Throughout \cref{sec:transcripts,sec:reduction,sec:switching}, we consider a
fixed query set $Q$ and a fixed subspace $K$. We use that $K$ is chosen at
random only in \cref{sec:theorem_proof}, to conclude the argument. We write
$\mu_s(\cdot\mid K)$ and $\rho(\cdot\mid K)$ for the conditional distributions
of the transcript given the subspace $K$. Similarly, we write
$\mu_s(\cdot\mid K,\pi)$ and $\rho(\cdot\mid K,\pi)$ for the conditional
distributions of the transcript given both $K$ and $\pi$ (where the randomness is over the choice of the functions $h$, $\bar{h}$.)

Let $F\subseteq[M]$ be the indices of the cosets of $K$ that contain a query
from $Q$, and $G\subseteq[M]$ those of $K^\perp$; each of these cosets
contains a query, so $|F|\leq q$ and $|G|\leq q$. For $i\in F$ and $j\in G$,
we fix a query $x_i$ in $C_i$ and a query $y_j$ in $D_j$ (a choice that
depends only on $Q$ and $K$) and call them the \emph{anchors} of $C_i$ and of
$D_j$. The answers $f(x_i)$ for $i\in F$ and $g(y_j)$ for $j\in G$ are the
\emph{anchor answers}.

Given $\pi$, the answers within a coset are determined by any one of them. For
every $i\in[M]$, all $x,x'\in C_i$ and all $y,y'\in D_{\pi(i)}$,
\begin{equation}\label{eq:coset_sums}
    f(x)+f(x')=\langle x+x',y\rangle
    \qquad\text{and}\qquad
    g(y)+g(y')=\langle x,y+y'\rangle.
\end{equation}
Indeed, adding~\eqref{eq:def_f} at $x$ and at $x'$ cancels their shared
random bit $h(i)$ and leaves $\langle x+x',\psi_{\pi(i)}\rangle$, which equals
$\langle x+x',y\rangle$ because $x+x'\in K$ and $\psi_{\pi(i)}+y\in K^\perp$.
The second identity follows in the same way from~\eqref{eq:def_g}, the bit
$\bar{h}(\pi(i))$ cancelling. Neither identity involves $h$, $\bar{h}$ or $s$,
so both hold whether the instance is drawn from $\mu_s$ or $\rho$.

\begin{definition}[Pinning graph]\label{def:pinning_graph}
    The \emph{pinning graph} $H_\tau$ of a transcript $\tau$ is the bipartite
    graph between the cosets of $K$ and the cosets of $K^\perp$ in which $C_i$
    and $D_j$ are adjacent if the answers recorded in $\tau$ satisfy
    $f(x)+f(x')=\langle x+x',y\rangle$ for all queries $x,x'\in C_i$ and all
    points $y\in D_j$, and $g(y)+g(y')=\langle x,y+y'\rangle$ for all queries
    $y,y'\in D_j$ and all points $x\in C_i$.
\end{definition}

We call $H_\tau$ the pinning graph because its perfect matchings are exactly
the permutations consistent with the answers that $\tau$ records within each
coset, much as the pinning interval of binary search contains exactly the
values consistent with the comparisons made so far.

By the first claim of \cref{lem:coset_pairing}, the right-hand sides in
\cref{def:pinning_graph} do not depend on which points of $D_j$ and of $C_i$
are used. We write $(i,j)$ for the edge between $C_i$ and $D_j$, and identify
a permutation $\pi$ with the perfect matching $\{(i,\pi(i)):i\in[M]\}$.
By~\eqref{eq:coset_sums}, every transcript $\tau$ that $\mu_s$ or $\rho$ produces
from $\pi$ has $\pi$ as a perfect matching in $H_\tau$. Importantly, the pinning
graph depends on $\tau$ only through the sums $f(x)+f(x')$ and $g(y)+g(y')$
within cosets, not through the anchor answers. Moreover, a coset containing at
most one query imposes no condition in \cref{def:pinning_graph}, as its only
pairs have $x=x'$ or $y=y'$; in particular, $(i,j)$ is an edge of $H_\tau$
whenever $i\notin F$ and $j\notin G$. \Cref{fig:pinning_graph} shows a
transcript and its pinning graph for the instance of \cref{fig:construction}.

\begin{figure}[t]
    \centering
    \small
    \begin{subfigure}[b]{0.26\textwidth}
        \centering
        \begin{tabular}{c@{\hspace{1.2em}}c@{\hspace{1.2em}}c}
            $x$ & $f(x)$ & $g(x)$\\
            \hline
            \rule{0pt}{2.6ex}%
            \texttt{0010} & 1 & 0\\
            \texttt{1110} & 0 & 1\\
            \texttt{1111} & 1 & 0\\
        \end{tabular}
        \caption{The transcript $\tau$}
    \end{subfigure}
    \hfill
    \begin{subfigure}[b]{0.32\textwidth}
        \centering
        \begin{tabular}{c@{\hspace{1.2em}}c@{\hspace{1.2em}}c}
            $x$ & \shortstack{coset\\of $K^{\phantom{\perp}}$}
                & \shortstack{coset\\of $K^\perp$}\\
            \hline
            \rule{0pt}{2.6ex}%
            \texttt{0010} & $C_2$ & $D_3$\\
            \texttt{1110} & $C_1$ & $D_3$\\
            \texttt{1111} & $C_1$ & $D_4$\\
        \end{tabular}
        \caption{The cosets of the queries}
    \end{subfigure}
    \hfill
    \begin{subfigure}[b]{0.30\textwidth}
        \centering
        \begin{tikzpicture}[x=0.6cm,y=-0.7cm,
                dot/.style={circle,fill,inner sep=1.1pt},
                lab/.style={font=\scriptsize,inner sep=1.5pt}]
            \foreach \i in {1,...,4} {
                \node[dot,label={[lab]left:$C_\i$}] (C\i) at (0,\i) {};
                \node[dot,label={[lab]right:$D_\i$}] (D\i) at (3,\i) {};
            }
            \foreach \i/\j in {1/2,1/4, 2/1,2/2,2/4, 3/1,3/2,3/3,3/4,
                    4/1,4/2,4/3,4/4}
                \draw[gray!70] (C\i) -- (D\j);
            \foreach \i/\j in {1/2, 2/4, 3/1, 4/3}
                \draw[thick] (C\i) -- (D\j);
        \end{tikzpicture}
        \caption{The pinning graph $H_\tau$}
    \end{subfigure}
    \caption{A transcript, the corresponding cosets and its pinning graph for $Q = \{ \text{\texttt{0010}, \texttt{1110}, \texttt{1111}} \}$ on the instance of
    \cref{fig:construction}. The algorithm gets access only to the transcript, and not the cosets or the pinning graphs. Only $C_1$ and $D_3$ hold two queries, so
    only these two cosets constrain $H_\tau$: the queries in $C_1$ have
    $x+x'=\texttt{0001}$ and $f(x)+f(x')=1$, so $C_1$ is adjacent only to $D_2$
    and $D_4$, and those in $D_3$ have $y+y'=\texttt{1100}$ and
    $g(y)+g(y')=1$, so $D_3$ is adjacent only to $C_3$ and $C_4$. The edges in
    bold are $\pi$, one of the eight perfect
    matchings of $H_\tau$ among the $24$ permutations of $[4]$.}%
    \label{fig:pinning_graph}
\end{figure}

We now consider the (partial) matching between $F$ and $G$ given by
\begin{equation}
    E(\pi)=\{(i,\pi(i)):i\in F,\ \pi(i)\in G\},
\end{equation}
which corresponds to the queried cosets of $K$ that $\pi$ pairs with queried
cosets of $K^\perp$.

Under $\mu_s$, a transcript produced from $\pi$ satisfies one more constraint
for each edge of $E(\pi)$: an edge $(i,j)$ has $x_i\in C_i$ and
$y_j\in D_{\pi(i)}$, so \cref{lem:coset_rep_independence} gives
$g(y_j)=f(x_i)+\langle x_i,y_j\rangle+s$, that is,
\begin{equation}\label{eq:edge_equation}
    f(x_i)+g(y_j)=\langle x_i,y_j\rangle+s
    \qquad\text{for every }(i,j)\in E(\pi).
\end{equation}
Under $\rho$, by contrast, $\bar{h}$ is independent of $h$, so each of these
equations holds with probability $1/2$, independently of the others.

\begin{definition}[Compatibility]\label{def:compatible}
    Let $\pi$ be a permutation of $[M]$, $s\in\F_2$ and let $\tau$ be a
    transcript. We write $\pi\bowtie\tau$ if $\pi$ is a perfect matching of
    $H_\tau$, and $(\pi,s)\bowtie\tau$ if moreover $\tau$
    satisfies~\eqref{eq:edge_equation}, and we then say that $\pi$,
    respectively $(\pi,s)$, is \emph{compatible} with $\tau$. For a bipartite
    graph $H$ between the cosets of $K$ and of $K^\perp$ we similarly write
    $\pi\bowtie H$ if $\pi$ is a perfect matching of $H$.
\end{definition}
Compatibility is exactly the condition for $\tau$ to be a possible transcript:
as \cref{lem:rho_given_pi,lem:mu_given_pi} show, $\rho(\tau\mid K,\pi)>0$ if,
and only if, $\pi\bowtie\tau$, and $\mu_s(\tau\mid K,\pi)>0$ if, and only if,
$(\pi,s)\bowtie\tau$.

\begin{lemma}\label{lem:rho_given_pi}
    Fix a subspace $K$ and permutation $\pi$ of $[M]$, and consider the distribution
    $\rho(\cdot\mid K,\pi)$. Then, every transcript in its support has the same
    pinning graph, which is uniquely determined by $K$ and $\pi$. Moreover, the
    anchor answers are uniform on $\F_2^{|F|+|G|}$, and each of their
    $2^{|F|+|G|}$ values extends to exactly one transcript in the support.
    Consequently, for every transcript $\tau$,
    \begin{equation}\label{eq:rho_given_pi}
        \rho(\tau\mid K,\pi)=2^{-|F|-|G|}\,\mathbbm{1}_{\pi\bowtie\tau}.
    \end{equation}
\end{lemma}

\begin{proof}
    Let $\tau$ be a transcript in the support. By~\eqref{eq:coset_sums} its
    sums within cosets are determined by $K$ and $\pi$, and the pinning graph
    depends on $\tau$ only through those sums, so $H_\tau$ is the same for
    every such $\tau$. Those sums
    also recover the remaining $2q-|F|-|G|$ answers of $\tau$ from its anchor
    answers, so no two transcripts in the support share the same anchor answers.

    By~\eqref{eq:def_f} and~\eqref{eq:def_g} the anchor answers are
    \begin{equation*}
        f(x_i)=h(i)+\langle x_i,\psi_{\pi(i)}\rangle,
        \qquad
        g(y_j)=\bar{h}(j)+\langle\phi_{\pi^{-1}(j)},y_j\rangle.
    \end{equation*}
    Under $\rho$ the bits $\{h(i)\}_{i\in F}$ and $\{\bar{h}(j)\}_{j\in G}$ are
    uniform over $\F_2^{|F|+|G|}$, so every assignment of anchor answers is
    possible and equally likely.

    By \cref{def:pinning_graph}, the answers $\tau$ records within its cosets
    obey the identities of~\eqref{eq:coset_sums} exactly when $\pi$ is a
    perfect matching of $H_\tau$, that is, when $\pi\bowtie\tau$, so the
    support consists of these transcripts, each of probability
    $2^{-|F|-|G|}$ by the two paragraphs above. Therefore
    $\rho(\tau\mid K,\pi)=2^{-|F|-|G|}\,\mathbbm{1}_{\pi\bowtie\tau}$.
\end{proof}

\begin{lemma}\label{lem:mu_given_pi}
    Fix a subspace $K$, $s\in\F_2$ and a permutation $\pi$ of $[M]$, and consider the
    distribution $\mu_s(\cdot\mid K,\pi)$. Every transcript in its support has
    the same pinning graph, which is uniquely determined by $K$ and $\pi$.
    The anchor answers are uniform on the values that
    satisfy~\eqref{eq:edge_equation}, and each of those
    $2^{|F|+|G|-|E(\pi)|}$ values extends to exactly one transcript in the
    support. Consequently, for every transcript $\tau$,
    \begin{equation}\label{eq:mu_given_pi}
        \mu_s(\tau\mid K,\pi)
        =2^{-|F|-|G|+|E(\pi)|}\,\mathbbm{1}_{(\pi,s)\bowtie\tau}.
    \end{equation}
\end{lemma}
\begin{proof}
    Let $\tau$ be a transcript in the support. The answers it records obey the
    identities of~\eqref{eq:coset_sums}, so $\pi\bowtie\tau$ and,
    by~\cref{lem:rho_given_pi}, $\tau$ lies in the support of
    $\rho(\cdot\mid K,\pi)$. Its pinning graph is therefore the one guaranteed by that lemma, and no two transcripts in the support share the same anchor
    answers.

    By~\eqref{eq:def_f} and~\eqref{eq:def_g} the anchor answers are again
    $f(x_i)=h(i)+\langle x_i,\psi_{\pi(i)}\rangle$ and
    $g(y_j)=\bar{h}(j)+\langle\phi_{\pi^{-1}(j)},y_j\rangle$. By~\eqref{eq:tie}, under $\mu_s$ the
    bits $\{h(i)\}_{i\in F}$ and $\{\bar{h}(j)\}_{j\in G\setminus\pi(F)}$ are the sum of
    distinct bits of $h$ and constant values (fully determined by $K,\pi$), so those anchor
    answers are uniform over $\F_2^{|F|+|G|-|E(\pi)|}$. As $E(\pi)$ is a
    matching, each of the remaining $|E(\pi)|$ anchor answers appears in
    exactly one equation of~\eqref{eq:edge_equation}, which determines it from
    the other. Every solution is therefore possible and equally likely, and
    there are $2^{|F|+|G|-|E(\pi)|}$ of them.

    By \cref{def:pinning_graph,def:compatible}, the answers $\tau$ records obey
    the identities
    of~\eqref{eq:coset_sums} and~\eqref{eq:edge_equation} exactly when
    $(\pi,s)\bowtie\tau$, so the support consists of these transcripts, each with
    probability $2^{-|F|-|G|+|E(\pi)|}$ by the two paragraphs above. Therefore
    $\mu_s(\tau\mid K,\pi)
    =2^{-|F|-|G|+|E(\pi)|}\,\mathbbm{1}_{(\pi,s)\bowtie\tau}$.
\end{proof}

Comparing~\eqref{eq:rho_given_pi} and~\eqref{eq:mu_given_pi},
$\mu_s(\cdot\mid K,\pi)$ is equal to $2^{|E(\pi)|}\,\rho(\cdot\mid K,\pi)$ on the
transcripts compatible with $(\pi,s)$ and zero elsewhere. Once $E(\pi)$ is
nonempty no transcript is compatible with both $(\pi,0)$ and $(\pi,1)$, so
$\mu_0(\cdot\mid K,\pi)$ and $\mu_1(\cdot\mid K,\pi)$ have disjoint supports (and thus easy to distinguish). The saving grace is that, \emph{averaged over $\pi$}, the two distributions will be close: indeed, 
\cref{sec:reduction,sec:switching} show that both are close to
$\rho(\cdot\mid K)$. To set the stage for these two sections, we must introduce some notation.\medskip

Under both $\mu_s$ and $\rho$ the permutation $\pi$ is uniform and independent
of $K$, so $\rho(\tau\mid K)$ and $\mu_s(\tau\mid K)$ are obtained by taking
the expectation of~\eqref{eq:rho_given_pi} and~\eqref{eq:mu_given_pi} over
$\pi$. In particular
$\rho(\tau\mid K)>0$ only if $H_\tau$ has a perfect matching.
For such a transcript $\tau$ we define the ratio
\begin{equation}
    L_s(\tau)=\frac{\mu_s(\tau\mid K)}{\rho(\tau\mid K)},
\end{equation}
which will play a crucial role in our lower bound, as the total variation
distance between $\mu_s(\cdot\mid K)$ and $\rho(\cdot\mid K)$ can be rewritten
as
\begin{align}
\| \mu_s(\cdot\mid K) - \rho(\cdot\mid K)\|_{\tv}
&= \frac{1}{2} \sum_{\tau}  | \mu_s(\tau\mid K) - \rho(\tau\mid K)| \nonumber\\
& = \frac{1}{2} \sum_{\tau: \rho(\tau\mid K) > 0} \rho(\tau\mid K)\cdot   | L_s(\tau) - 1| \nonumber\\
&= \frac{1}{2} \mathbb{E}_{\tau\sim\rho(\cdot\mid K)} [ | L_s(\tau)-1 | ]\,,
\label{eq:tv_as_ratio}
\end{align}
where we used that $\mu_s(\tau\mid K)>0$ implies $\rho(\tau\mid K)>0$ (as $(\pi,s)\bowtie\tau$ implies $\pi\bowtie\tau$) for the second equality.

We close this subsection by re-writing $L_s(\tau)$ in terms of random permutations:
\begin{equation}
    L_s(\tau)
    =\frac{\mathbb{E}_\pi\big[\mu_s(\tau\mid K,\pi)\big]}
           {\mathbb{E}_\pi\big[\rho(\tau\mid K,\pi)\big]}
    =\frac{\mathbb{E}_\pi\big[2^{|E(\pi)|}\,
        \mathbbm{1}_{(\pi,s)\bowtie\tau}\big]}
           {\Pr_\pi[\pi\bowtie\tau]}
    =\mathbb{E}_{\pi\bowtie\tau}\big[2^{|E(\pi)|}\,
        \mathbbm{1}_{(\pi,s)\bowtie\tau}\big].
    \label{eq:ratio}
\end{equation}
Here the first equality takes the expectation over $\pi$; the
second substitutes~\eqref{eq:rho_given_pi} and~\eqref{eq:mu_given_pi} and
cancels out the common factor
$2^{-|F|-|G|}$; and the third uses the fact that
$2^{|E(\pi)|}\,\mathbbm{1}_{(\pi,s)\bowtie\tau}$ vanishes unless
$\pi\bowtie\tau$, so its unconditional expectation, divided by the fraction of
$\pi$ with $\pi\bowtie\tau$, is the conditional expectation.

\subsection{Reduction to the Matching Probabilities}
\label{sec:reduction}
By~\cref{eq:tv_as_ratio,eq:ratio}, to bound the TV distance between transcripts (conditioned on a given $K$), it is enough to bound an expectation over transcripts and compatible permutations $\pi$. In this section, we take the first step to simplify that task, by reducing it to the study of matchings of the pinning graph. Specifically, for $I\subseteq F$, $J\subseteq G$ and a bipartite graph $H$ between the
cosets of $K$ and of $K^\perp$, let
\begin{equation}
    \label{eq:matching:probas}
    p_{I,J}(H)=\Pr\nolimits_{\pi\bowtie H}[\pi(I)=J]
\end{equation}
be the probability that a uniformly random perfect matching $\pi$ of $H$ maps
$I$ onto $J$. We call these the \emph{matching probabilities}. Since
$I\subseteq F$ and $J\subseteq G$, the event $\pi(I)=J$ says that $E(\pi)$
matches $I$ to $J$. The next lemma bounds the distance between
$\mu_s(\cdot\mid K)$ and $\rho(\cdot\mid K)$, through~\eqref{eq:tv_as_ratio},
by the matching probabilities of the pinning graph alone.
\begin{lemma}\label{lem:reduction}
    For every $s\in\F_2$,
    \begin{equation*}
        \mathbb{E}_{\tau\sim\rho(\cdot\mid K)}\big[|L_s(\tau)-1|\big]
        \leq\Big(\mathbb{E}_{\tau\sim\rho(\cdot\mid K)}\Big[
            \sum_{\substack{I\subseteq F,\,J\subseteq G\\|I|=|J|\geq1}}
            p_{I,J}(H_\tau)^{2}\Big]\Big)^{1/2}.
    \end{equation*}
\end{lemma}
\begin{proof}
    By Cauchy--Schwarz,
    \begin{equation}
        \label{eq:reduction:cs}
        \mathbb{E}_{\tau\sim\rho(\cdot\mid K)}\big[|L_s(\tau)-1|\big]
        \leq\Big(\mathbb{E}_{\tau\sim\rho(\cdot\mid K)}
             \big[ \big(L_s(\tau)-1\big)^{2}\big]\Big)^{1/2}.
    \end{equation}

    Fix a graph $H$ that is the pinning graph of some transcript with positive
    probability under $\rho(\cdot\mid K)$. By~\cref{lem:rho_given_pi} the
    pinning graph is determined by $K$ and $\pi$, so the anchor answers, being uniform given $K$ and $\pi$,
    remain uniform on $\F_2^{|F|+|G|}$ after conditioning on the event $\{H_\tau=H\}$.

    Squaring~\eqref{eq:ratio} and taking expectations over $\tau$, we get
    \begin{align*}
        \mathbb{E}_{\tau\sim\rho(\cdot\mid K)}\big[L_s(\tau)^{2}\,\big|\,H_\tau=H\big]
        &= \mathbb{E}_{\tau\sim\rho(\cdot\mid K)}\Big[ \mathbb{E}_{\pi,\pi'\bowtie H}\Big[2^{|E(\pi)|+|E(\pi')|} \mathbbm{1}_{(\pi,s)\bowtie\tau} \mathbbm{1}_{(\pi',s)\bowtie\tau}\Big]\,\big|\,H_\tau=H\Big]\\
        &=\mathbb{E}_{\pi,\pi'\bowtie H}\Big[2^{|E(\pi)|+|E(\pi')|}\,
            \Pr\nolimits_\tau\big[(\pi,s)\bowtie\tau,\ (\pi',s)\bowtie\tau
                \,\big|\,H_\tau=H\big]\Big].
    \end{align*}
    For fixed $\pi$ and $\pi'$ compatible with $H$, consider the multigraph on the queried cosets
    with an edge between $C_i$ and $D_j$ for each $(i,j)\in E(\pi)$ and one
    for each $(i,j)\in E(\pi')$. Every coset intersects at most one edge of each
    matching, so each connected component is either a path or a cycle. (If an edge is used in both matchings, we form a cycle of length two between the same two cosets). We illustrate this in \cref{fig:connected-components}. Let
    $\mathcal{C}(\pi,\pi')$ be the set of cycles.
\begin{figure}[t]
    \centering
    \small
    \begin{tikzpicture}[x=0.6cm,y=-0.7cm,
            dot/.style={circle,fill,inner sep=1.1pt},
            lab/.style={font=\scriptsize,inner sep=1.5pt},
            grp/.style={font=\scriptsize,anchor=west,inner sep=1.5pt}]
        \node[dot,label={[lab]left:$C_1$}]  (C1) at (0,0) {};
        \node[dot,label={[lab]left:$C_2$}]  (C2) at (0,1) {};
        \node[dot,label={[lab]right:$D_1$}] (D1) at (3,0) {};
        \node[dot,label={[lab]right:$D_2$}] (D2) at (3,1) {};
        \draw[gray!70] (C2) -- (D1);
        \draw[thick] (C1) -- (D1);
        \draw[thick] (C2) -- (D2);
        \node[grp] at (4.1,0.5) {a path};

        \node[dot,label={[lab]left:$C_3$}]  (C3) at (0,2.2) {};
        \node[dot,label={[lab]left:$C_4$}]  (C4) at (0,3.2) {};
        \node[dot,label={[lab]right:$D_3$}] (D3) at (3,2.2) {};
        \node[dot,label={[lab]right:$D_4$}] (D4) at (3,3.2) {};
        \draw[gray!70] (C3) -- (D4);
        \draw[gray!70] (C4) -- (D3);
        \draw[thick] (C3) -- (D3);
        \draw[thick] (C4) -- (D4);
        \node[grp] at (4.1,2.7) {a cycle};

        \node[dot,label={[lab]left:$C_5$}]  (C5) at (0,4.4) {};
        \node[dot,label={[lab]right:$D_5$}] (D5) at (3,4.4) {};
        \draw[gray!70] (C5) to[bend right=25] (D5);
        \draw[thick] (C5) to[bend left=25] (D5);
        \node[grp] at (4.1,4.4) {a cycle of length two};

        \draw[thick] (9.5,1.7) -- (10.7,1.7);
        \node[lab,anchor=west] at (10.9,1.7) {edge in $E(\pi)$};
        \draw[gray!70] (9.5,2.7) -- (10.7,2.7);
        \node[lab,anchor=west] at (10.9,2.7) {edge in $E(\pi')$};
    \end{tikzpicture}
    \caption{Connected components of the multigraph on the queried cosets,
    which carries one edge for each pair in $E(\pi)$ and one for each pair
    in $E(\pi')$. Every coset meets at most one edge of each matching, so a
    component is a path or a cycle. The last component comes from a pair
    that lies in both matchings and so contributes two parallel edges.}%
    \label{fig:connected-components}
\end{figure}

    Since $\pi$ and $\pi'$ are perfect matchings of $H=H_\tau$, by
    \cref{def:compatible} the two compatibilities in this probability hold
    exactly when the anchor answers satisfy~\eqref{eq:edge_equation} for every
    $(i,j)\in E(\pi)\cup E(\pi')$. This constitutes a system of linear
    equations over $\F_2$ in the $|F|+|G|$ anchor answers, where the equation
    of an edge involves the two answers of its endpoints. To bound its rank, note that deleting one edge from each cycle leaves a forest, whose equations are now
    linearly independent: the system therefore has rank at least
    $|E(\pi)|+|E(\pi')|-|\mathcal{C}(\pi,\pi')|$, and since the anchor answers
    are uniform these equations simultaneously hold with probability at most
    $2^{-|E(\pi)|-|E(\pi')|+|\mathcal{C}(\pi,\pi')|}$. Thus,
    \[
    \Pr\nolimits_\tau\big[(\pi,s)\bowtie\tau,\ (\pi',s)\bowtie\tau
                \,\big|\,H_\tau=H\big]
                \leq 2^{-|E(\pi)|-|E(\pi')|+|\mathcal{C}(\pi,\pi')|}
    \]
    and so
    $\mathbb{E}_\tau[L_s(\tau)^{2}\mid H_\tau=H]\leq
    \mathbb{E}_{\pi,\pi'\bowtie H}[2^{|\mathcal{C}(\pi,\pi')|}]$.
    Similarly, taking the expectation of~\eqref{eq:ratio} over $\tau$ and exchanging it
    with the expectation over $\pi$ gives
    \begin{equation*}
        \mathbb{E}_\tau\big[L_s(\tau)\,\big|\,H_\tau=H\big]
        =\mathbb{E}_{\pi\bowtie H}\Big[2^{|E(\pi)|}\,
            \Pr\nolimits_\tau\big[(\pi,s)\bowtie\tau
                \,\big|\,H_\tau=H\big]\Big],
    \end{equation*}
    and $E(\pi)$ is itself a forest, so its equations are linearly independent
    and consistent: they hold with probability exactly $2^{-|E(\pi)|}$, which
    cancels the factor $2^{|E(\pi)|}$ and yields
    $\mathbb{E}_\tau[L_s(\tau)\mid H_\tau=H]=1$.

    With these two (in)equalities in hand, by expanding the square we get
    \begin{align}
        \mathbb{E}_{\tau}\big[(L_s(\tau)-1)^{2}\,\big|\,H_\tau=H\big]
        &=\mathbb{E}_{\tau}\big[L_s(\tau)^{2}\,\big|\,H_\tau=H\big]
         -2\,\mathbb{E}_{\tau}\big[L_s(\tau)\,\big|\,H_\tau=H\big]+1
         \nonumber\\
        &=\mathbb{E}_{\tau}\big[L_s(\tau)^{2}\,\big|\,H_\tau=H\big]-1
        \nonumber\\
        & \leq\mathbb{E}_{\pi,\pi'\bowtie H}
              \big[2^{|\mathcal{C}(\pi,\pi')|}\big]-1.
        \label{eq:cycle_moment}
    \end{align}
    To continue, fix two perfect matchings $\pi$ and $\pi'$ of $H$. We define a bijection between nonempty $\mathcal{S}\subseteq\mathcal{C}(\pi,\pi')$
    and pairs $(I, J)$ with $\emptyset \neq I \subseteq F$, $J \subseteq G$,
    $\pi(I) = \pi'(I) = J$ as follows: for a given nonempty $\mathcal{S}$
    we let $I$ and $J$ index the cosets covered by the cycles in
    $\mathcal{S}$, which satisfy $\pi(I) = \pi'(I) = J$ and
    $I \neq \emptyset$. One can check that this is a bijection.
    Therefore, for any $\pi,\pi'$,
    \begin{equation*}
        2^{|\mathcal{C}(\pi,\pi')|}-1
        =\big|\{\mathcal{S}\subseteq\mathcal{C}(\pi,\pi')
            :\mathcal{S}\neq\emptyset\}\big|
        =\sum_{\substack{I\subseteq F,\,J\subseteq G\\|I|=|J|\geq1}}
            \mathbbm{1}_{\pi(I)=\pi'(I)=J}.
    \end{equation*}
    Taking expectations over $\pi,\pi'$ compatible with $H$,
    \begin{equation*}
        \mathbb{E}_{\pi,\pi'\bowtie H}\big[2^{|\mathcal{C}(\pi,\pi')|}\big]-1
        =\sum_{\substack{I\subseteq F,\,J\subseteq G\\|I|=|J|\geq1}}
            \Pr\nolimits_{\pi\bowtie H}[\pi(I)=J]\,
            \Pr\nolimits_{\pi'\bowtie H}[\pi'(I)=J]
        =\sum_{\substack{I\subseteq F,\,J\subseteq G\\|I|=|J|\geq1}}
            p_{I,J}(H)^{2}\,,
    \end{equation*}
    where the first equality follows from %
    the  independence of $\pi$ and $\pi'$, and the second from the
    definition of $p_{I,J}(H)$, since $\pi$ and $\pi'$ are uniform perfect
    matchings of $H$.
    Averaging~\eqref{eq:cycle_moment} over the pinning graph $H$ of
    $\tau\sim\rho(\cdot\mid K)$ and plugging the result in~\eqref{eq:reduction:cs} establishes the lemma.
\end{proof}

\subsection{Bounding the Matching Probabilities}
\label{sec:switching}
In view of~\cref{lem:reduction}, we will bound the matching probabilities $p_{I,J}(H)$ via the probability
that~$E(\pi)$ contains a given matching $T\subseteq F\times G$, for a
uniformly random perfect matching $\pi$ of $H$. The goal is to show that such a probability is low for any matching $T$, unless the number of queries is large. The direct consequence is that small transcripts are not sufficient to even pin down one edge in any matching, the intuition behind (recalling~\cref{lem:pin_sum}) that this means the adversary cannot deduce $s$ with high probability. We start by noting that a uniformly sampled permutation pairs a given $i\in F$ with a given $j\in G$
with probability $1/M$, and thus contains a matching $T$ with probability at least $1/M^{|T|}$, where $|T|$ is the number of edges in the matching. Our next lemma shows that, when $K$ is $r$-good and $q$ is not too large, this na\"ive lower bound is not too far from the truth: namely, a
uniformly random perfect matching of a pinning graph contains $T$ with
probability at most $(4^{r+1}/M)^{|T|}$:

\begin{lemma}\label{lem:switching}
    Let $K$ be $r$-good, and $H$ be the pinning graph of a transcript $\tau$
    with $\rho(\tau\mid K)>0$. Suppose further that $q\leq M/2^{r+3}$. Then, for every
    matching $T\subseteq F\times G$,
    \begin{equation*}
        \Pr\nolimits_{\pi\bowtie H}[T\subseteq E(\pi)]
        \leq\big(4^{r+1}/M\big)^{|T|}.
    \end{equation*}
\end{lemma}
\begin{proof}
Fix $K$, $H$, and $\tau$ as in the statement of the lemma. 
    We first show that, as long as $q$ is small enough, each queried endpoint participates in many edges of $H$ with the other endpoint unqueried: 
    \begin{claim}
        \label{lem:switching:claim1}
    For $K$, $H$, and $\tau$ as above, we have
    \begin{equation}\label{eq:degree_bound}
        \begin{aligned}
            \big|\{j\notin G:(i,j)\in H\}\big|&\geq M/2^{r}-q
                &&\quad\text{for all }i\in F,\\
            \big|\{i\notin F:(i,j)\in H\}\big|&\geq M/2^{r}-q
                &&\quad\text{for all }j\in G.
        \end{aligned}
    \end{equation}
    \end{claim}
    \begin{proof}
    Fix $i\in F$ and let $W\subseteq K$ be the affine span of $Q\cap C_i$, which has dimension at most $r$ by our assumption that $K$ is
    $r$-good. As noted in \cref{sec:transcripts}, the fact that $\rho(\tau\mid K)>0$ guarantees that
    $H$ has (at least) one perfect matching; fix such a matching, and let $j_0$ be the index it matches index $i$ to. As
    $(i,j_0)$ is an edge of $H$, for all queries $x,x'\in C_i$ we have
    $f(x)+f(x')=\langle x+x',\psi_{j_0}\rangle$.
    For an index $j\notin G$, the coset $D_j$ is unqueried, so $(i,j)$ is an
    edge of $H$ exactly when $f(x)+f(x')=\langle x+x',\psi_j\rangle$ for those
    same $x,x'$: that is, when
    $\langle w,\psi_j\rangle=\langle w,\psi_{j_0}\rangle$ for every $w\in W$.
    Applying (the third item of)~\cref{lem:coset_pairing} to a basis of $W$ we get exactly
    $M/2^{\dim W}\geq M/2^{r}$ indices $j$ in $[M]$ satisfying this equation,
    and since $|G|\leq q$, at least $M/2^{r}-q$ of them lie outside $G$ and so
    give edges of $H$.
    This gives the first bound in~\eqref{eq:degree_bound}: the second is obtained analogously, repeating the same argument with $K$ and $K^\perp$ interchanged.
    \end{proof}

    Let $\Pi$ be the set of perfect matchings of $H$. For a matching
    $T\subseteq F\times G$, let $\Pi_T\subseteq \Pi$ be the subset of perfect matchings that contain every edge of
    $T$; since $T\subseteq F\times G$, these are exactly the $\pi$'s with
    $T\subseteq E(\pi)$. We next prove that adding a new edge to $T$ 
    noticeably shrinks the size of $\Pi_T$:
    \begin{claim}
        \label{lem:switching:claim2}
    Assume $1\leq q\leq M/2^{r+3}$. Fix any $T\subseteq F\times G$, and any edge
    $(i_1,j_1)\in (F\times G)\setminus T$. Then
    \begin{equation}\label{eq:one_edge}
        \big|\Pi_{T\cup\{(i_1,j_1)\}}\big|
        \leq\frac{4^{r+1}}{M}\,\big|\Pi_T\big|.
    \end{equation}
    \end{claim}
    \begin{proof}
    Consider any $\pi\in\Pi_{T\cup\{(i_1,j_1)\}}$. We construct an alternating cycle $S$
    of $\pi$ containing $(i_1,j_1)$, involving no queried coset other than
    $C_{i_1}$ and $D_{j_1}$.
    Choose $i_2\notin F\cup\pi^{-1}(G)$ with $(i_2,j_1)\in H$ and let
    $j_2=\pi(i_2)$; then choose $j_3\notin G\cup\pi(F)$ with $j_3\neq j_2$
    and $(i_1,j_3)\in H$, and let $i_3=\pi^{-1}(j_3)$. Such indices exist
    by~\cref{lem:switching:claim1}: since $|\pi(F)|,|\pi^{-1}(G)|\leq q$, the bounds
    in~\eqref{eq:degree_bound} guarantee at least $M/2^{r}-2q$ choices of $i_2$,
    and at least $M/2^{r}-2q-1$ choices of $j_3$ once $j_2$ is forbidden. Our assumption on $q$ implies 
    $M/2^{r}-2q-1\geq M/2^{r}-4q\geq M/2^{r+1}$, so $\pi$ admits at least
    $(M/2^{r+1})^{2}$ such cycles.
    
    Note that $i_2\notin F$ and $j_3\notin G$ by
    choice, and that $i_2\notin\pi^{-1}(G)$ and $j_3\notin\pi(F)$ imply that $j_2\notin G$ and $i_3 \notin F$, so none of $C_{i_2}$,
    $C_{i_3}$, $D_{j_2}$, $D_{j_3}$ contains a query. Our alternating cycle
    $S$ is then
    \begin{equation*}
        i_1\overset{\pi}{\to} j_1\to i_2\overset{\pi}{\to} j_2\to i_3\overset{\pi}{\to} j_3\to i_1,
    \end{equation*}
    as shown in \cref{fig:three-edge-switch}. Note that $S\subseteq H$: the
    three edges of $\pi$ because $\pi\in\Pi$, the edges $(i_2,j_1)$ and
    $(i_1,j_3)$ by the choice of $i_2$ and $j_3$, and $(i_3,j_2)$ because
    neither $C_{i_3}$ nor $D_{j_2}$ holds a query, so
    \cref{def:pinning_graph} puts no condition on it.

    Let $\pi'=\pi\mathbin{\triangle}S$ be the symmetric difference of $\pi$
    and $S$.
    Since $S$ alternates with $\pi$ and lies in $H$, $\pi'\in\Pi$, and
    $\pi'\mathbin{\triangle}S=\pi$. Moreover $T\cap S=\emptyset$, so $\pi'$
    contains every edge of $T$.

    \begin{figure}[htpb!]
    \centering
    \small
    \begin{subfigure}[b]{0.4\textwidth}
        \centering
        \begin{tikzpicture}[x=0.65cm,y=-1.0cm,
                dot/.style={circle,fill,inner sep=1.1pt},
                lab/.style={font=\scriptsize,inner sep=1.5pt}]
            \node[dot,label={[lab]left:$i_1$}]  (I1) at (0,0) {};
            \node[dot,label={[lab]left:$i_2$}]  (I2) at (0,1) {};
            \node[dot,label={[lab]left:$i_3$}]  (I3) at (0,2) {};
            \node[dot,label={[lab]right:$j_1$}] (J1) at (3,0) {};
            \node[dot,label={[lab]right:$j_2$}] (J2) at (3,1) {};
            \node[dot,label={[lab]right:$j_3$}] (J3) at (3,2) {};
            \draw[gray!70] (I2) -- (J1);
            \draw[gray!70] (I3) -- (J2);
            \draw[gray!70] (I1) -- (J3);
            \draw[thick] (I1) -- (J1);
            \draw[thick] (I2) -- (J2);
            \draw[thick] (I3) -- (J3);
        \end{tikzpicture}
        \caption{Before the switch}
    \end{subfigure}
    \hspace{2em}
    \begin{subfigure}[b]{0.4\textwidth}
        \centering
        \begin{tikzpicture}[x=0.65cm,y=-1.0cm,
                dot/.style={circle,fill,inner sep=1.1pt},
                lab/.style={font=\scriptsize,inner sep=1.5pt}]
            \node[dot,label={[lab]left:$i_1$}]  (I1) at (0,0) {};
            \node[dot,label={[lab]left:$i_2$}]  (I2) at (0,1) {};
            \node[dot,label={[lab]left:$i_3$}]  (I3) at (0,2) {};
            \node[dot,label={[lab]right:$j_1$}] (J1) at (3,0) {};
            \node[dot,label={[lab]right:$j_2$}] (J2) at (3,1) {};
            \node[dot,label={[lab]right:$j_3$}] (J3) at (3,2) {};
            \draw[gray!70] (I1) -- (J1);
            \draw[gray!70] (I2) -- (J2);
            \draw[gray!70] (I3) -- (J3);
            \draw[thick] (I2) -- (J1);
            \draw[thick] (I3) -- (J2);
            \draw[thick] (I1) -- (J3);
        \end{tikzpicture}
        \caption{After the switch}
    \end{subfigure}

    \caption{The switch in the proof of %
    \cref{lem:switching:claim2}. All six pairs
    drawn are edges of $H$, and together they form a cycle of length six; the
    three in bold are the edges of the matching, $\pi$ on the left and $\pi'$
    on the right. Here $i_1\in F$ and $j_1\in G$, while $i_2$ and $i_3$ lie
    outside $F$ and $j_2$ and $j_3$ outside $G$.}%
    \label{fig:three-edge-switch}
    \end{figure}

    Consider now all pairs $(\pi,S)$ with $\pi\in\Pi_{T\cup\{(i_1,j_1)\}}$ and $S$ an
    alternating cycle for $\pi$ as constructed above: by the foregoing discussion, there are at least
    $(M/2^{r+1})^{2}\,|\Pi_{T\cup\{(i_1,j_1)\}}|$ of them. We define the map which takes every such $(\pi,S)$ to the corresponding pair $(\pi',j_2)$ (as obtained during the construction). This is injective, since $S$ can be recovered from it (and then $\pi$ from $S,\pi'$, as $\pi=\pi'\mathbin{\triangle}S$): indeed, recall that 
    $j_3=\pi'(i_1)$, $i_2=\pi'^{-1}(j_1)$ and $i_3=\pi'^{-1}(j_2)$. But the image of this map has cardinality (at most) $M\cdot |\Pi_T|$, since there 
    are at most $|\Pi_T|$ choices for $\pi'$ and at most $M$ choices for $j_2$. 
    Putting this all together, we get
    \[
        (M/2^{r+1})^{2}\,|\Pi_{T\cup\{(i_1,j_1)\}}|\leq M\,|\Pi_T|\,,
    \]
    which establishes the claim.
    \end{proof}
    To conclude the proof of~\cref{lem:switching}, write $T = \{e_1,\dots,e_k\}$ and for $0\leq \ell \leq k = |T|$ let $T_\ell$
    be the subset of the first $\ell$ edges, so that $T_0=\emptyset$ and
    $\Pi_{T_0}=\Pi$. Since $T$ is a matching, $e_\ell$
    does not intersect any edge of $T_{\ell-1}$, so by~\cref{lem:switching:claim2} we have
    $|\Pi_{T_\ell}|\leq(4^{r+1}/M)\,|\Pi_{T_{\ell-1}}|$, and hence
    $|\Pi_T|\leq(4^{r+1}/M)^{k}|\Pi|$. The conclusion follows, as
    $\Pr\nolimits_{\pi\bowtie H}[T\subseteq E(\pi)] = |\Pi_T|/|\Pi|$ (recall that $\Pi \neq \emptyset$ as $\rho(\tau\mid K)>0$ by assumption).
\end{proof}

The next lemma builds on \cref{lem:switching} to derive a bound on the sum of the
squared matching probabilities, allowing us to handle the RHS of \cref{lem:reduction}.

\begin{lemma}\label{lem:pin_sum}
    Let $K$ be $r$-good, and $H$ be the pinning graph of a transcript $\tau$
    with $\rho(\tau\mid K)>0$. Suppose further that $q\leq M/2^{r+3}$ and that $\beta=|F|\,|G|\,(4^{r+1}/M)^{2} < 1 $. Then,
    \[
            \sum_{\substack{I\subseteq F,\,J\subseteq G\\|I|=|J|\geq1}}
            p_{I,J}(H)^{2} \leq \frac{\beta}{1-\beta}\,.
    \]
\end{lemma}
\begin{proof}
    Fix $k\geq1$, and any $I\subseteq F,\,J\subseteq G$ with $|I|=|J|=k$. The event $\{\pi(I)=J\}$ can be written as the union, over every bijection $\sigma\colon I\to J$, of the event that
    $E(\pi)$ contains every edge $(i,\sigma(i))$: so \cref{lem:switching} and a
    union bound give $p_{I,J}(H)\leq k!\,(4^{r+1}/M)^{k}$. Squaring and summing
    over all subsets of size $k$,
    \begin{equation*}
        \sum_{\substack{I\subseteq F,\,J\subseteq G\\|I|=|J| = k}}
            p_{I,J}(H)^{2}
        \leq\binom{|F|}{k}\binom{|G|}{k}(k!)^{2}\left(\frac{4^{r+1}}{M}\right)^{2k}
        \leq\left(|F|\,|G|\,\left(\frac{4^{r+1}}{M}\right)^{2}\right)^{k}=\beta^{k}\,,
    \end{equation*}
    using $\binom{\ell}{k}\leq \frac{\ell^{k}}{k!}$. 
    Summing $\beta^{k}$ over all $k\geq1$ gives $\beta/(1-\beta)$, and the lemma.
\end{proof}

\subsection{Putting it together: Proof of~\cref{thm:main}}
We finally have all the ingredients to establish~\cref{thm:main}, and thus~\cref{thm:main:forrelation}.
\label{sec:theorem_proof}

\begin{proof}[Proof of \cref{thm:main}]
    Fix an integer $1\leq r<m$. The total variation distance is at most $1$, so if
    $q>M/(2\cdot4^{r+1})$ the second term of~\eqref{eq:main_bound} already
    exceeds $1$ and there is nothing to prove. We thus hereafter assume $q\leq M/(2\cdot4^{r+1})$; in particular $q\leq M/2^{r+3}$, letting us invoke
    \cref{lem:switching,lem:pin_sum}. By convexity of total variation distance,
    \begin{align*}
        \|\mu_0-\mu_1\|_{\tv}
        &\leq\mathbb{E}_K[ \|\mu_0(\cdot\mid K)-\mu_1(\cdot\mid K)\|_{\tv} ]
        \tag{Convexity of TV distance} \\ 
        &\leq\Pr[K\text{ is not $r$-good}]
            +\max_{K\text{ $r$-good}}
             \|\mu_0(\cdot\mid K)-\mu_1(\cdot\mid K)\|_{\tv} \\
        &\leq 2\binom{q}{r+2}M^{-(r+1)}  +\max_{K\text{ $r$-good}}
             \|\mu_0(\cdot\mid K)-\mu_1(\cdot\mid K)\|_{\tv} \tag{\Cref{lem:good_subspace}}\,.
    \end{align*}
    We bound the second term as follows: fixing any $r$-good $K$, we have
 by the triangle inequality
 \[
 \|\mu_0(\cdot\mid K)-\mu_1(\cdot\mid K)\|_{\tv}
 \leq \|\mu_0(\cdot\mid K)-\rho(\cdot\mid K)\|_{\tv}
 + \|\mu_1(\cdot\mid K)-\rho(\cdot\mid K)\|_{\tv}\,.
 \]
    Since $|F|,|G|\leq q \leq  M/(2\cdot4^{r+1})$, the setting of $\beta$ from \cref{lem:pin_sum} satisfies
    \[
    \beta=|F||G|\left(\frac{4^{r+1}}{M}\right)^{2} 
    \leq \left(\frac{4^{r+1}q}{M}\right)^{2} 
    \leq \frac{1}{4} < 1\,.
    \]
    So we can apply~\cref{lem:pin_sum}, which, since $\beta/(1-\beta)\leq 2\beta$, combined with \cref{lem:reduction} and~\eqref{eq:tv_as_ratio} gives that
    \begin{equation*}
        \|\mu_s(\cdot\mid K)-\rho(\cdot\mid K)\|_{\tv}
        \leq\tfrac12\sqrt{2\beta}\leq\frac{4^{r+1}q}{M}
    \end{equation*}
    for both $s\in \{0,1\}$. Along with the bound on the
    first term, this proves~\eqref{eq:main_bound}.

    Turning to the ``in particular'' statement, every instance of \cref{sec:construction} is a
    \yes-instance when $s=0$ and a \no-instance when $s=1$, by
    \cref{lem:forr_value}. So any algorithm that solves extremal Forrelation
    with probability at least $2/3$ must distinguish $\mu_0$ from $\mu_1$ with probability at least
    $2/3$, which is only possible if $\|\mu_0-\mu_1\|_{\tv}\geq 1/3$. 
    
    Keeping $1\leq r < m$ as a free parameter, assume $q \leq M^{1-\frac{1}{r+2}}$. The bound in~\eqref{eq:main_bound} implies
    \[
        \|\mu_0-\mu_1\|_{\tv}
        \leq 2\cdot \frac{q^{r+2}}{M^{r+1}} + 2\cdot \frac{4^{r+1}}{M^{\frac{1}{r+2}}}
        \leq \frac{2}{2^m} + 2^{2r+3 - \frac{m}{r+2}}\,.
    \]
    The first term is $o(1)$, and, setting $r = \sqrt{m}/2$, the second becomes $2^{-\Omega(\sqrt{m})}=o(1)$ as well. Thus, for this choice of $r$, any successful (non-adaptive) algorithm
    must have query complexity 
    \[
    q > M^{1-\frac{1}{r+2}} = 2^{\frac{n}{2} - O(\sqrt{n})}
    \]
    recalling that $m=n/2$. This is $2^{(1/2-o(1))n}$ with the $o(1)$ term
    of order $O(1/\sqrt{n})$.
\end{proof}

\bibliographystyle{alpha}
\bibliography{ref}

\begin{thebibliography}{YHLW21}

\bibitem[AA18]{DBLP:journals/siamcomp/AaronsonA18}
Scott Aaronson and Andris Ambainis.
\newblock Forrelation: {A} problem that optimally separates quantum from
  classical computing.
\newblock {\em {SIAM} J. Comput.}, 47(3):982--1038, 2018.

\bibitem[Aar10]{DBLP:conf/stoc/Aaronson10}
Scott Aaronson.
\newblock {BQP} and the polynomial hierarchy.
\newblock In Leonard~J. Schulman, editor, {\em Proceedings of the 42nd {ACM}
  Symposium on Theory of Computing, {STOC} 2010, Cambridge, Massachusetts, USA,
  5-8 June 2010}, pages 141--150. {ACM}, 2010.

\bibitem[BGGS21]{Bravyi2021}
Sergey Bravyi, David Gosset, Daniel Grier, and Luke Schaeffer.
\newblock Classical algorithms for forrelation, 2021.

\bibitem[CCM26]{CanonneCM26v1}
Cl{\'{e}}ment~L. Canonne, Kenny Chen, and Juli{\'{a}}n Mestre.
\newblock The quantumly fast and the classically forrious.
\newblock {\em CoRR}, abs/2602.07503v1, 2026.
\newblock Version 1.

\bibitem[Dil75]{dillon1974elementary}
J.~F. Dillon.
\newblock Elementary {H}adamard difference sets.
\newblock In {\em Proceedings of the {S}ixth {S}outheastern {C}onference on
  {C}ombinatorics, {G}raph {T}heory and {C}omputing ({F}lorida {A}tlantic
  {U}niv., {B}oca {R}aton, {F}la., 1975)}, Congress. Numer., No. XIV, pages
  237--249. Utilitas Math., Winnipeg, MB, 1975.

\bibitem[Gro96]{DBLP:conf/stoc/Grover96}
Lov~K. Grover.
\newblock A fast quantum mechanical algorithm for database search.
\newblock In Gary~L. Miller, editor, {\em Proceedings of the Twenty-Eighth
  Annual {ACM} Symposium on the Theory of Computing, Philadelphia,
  Pennsylvania, USA, May 22-24, 1996}, pages 212--219. {ACM}, 1996.

\bibitem[GS26]{GirishS26}
Uma Girish and Rocco~A. Servedio.
\newblock Forrelation is extremally hard.
\newblock In {\em {ITCS}}, volume 362 of {\em LIPIcs}, pages 72:1--72:22.
  Schloss Dagstuhl - Leibniz-Zentrum f{\"{u}}r Informatik, 2026.

\bibitem[Ham25]{hamoudi2025briefintroductionquantumquery}
Yassine Hamoudi.
\newblock A brief introduction to quantum query complexity, 2025.

\bibitem[McF73]{McFarland73}
Robert~L McFarland.
\newblock A family of difference sets in non-cyclic groups.
\newblock {\em Journal of Combinatorial Theory, Series A}, 15(1):1--10, 1973.

\bibitem[MdW16]{DBLP:journals/toc/MontanaroW16}
Ashley Montanaro and Ronald de~Wolf.
\newblock A survey of quantum property testing.
\newblock {\em Theory Comput.}, 7:1--81, 2016.

\bibitem[PPZK26]{PasalicPZK26}
Enes Pasalic, Alexandr Polujan, Fengrong Zhang, and Sadmir Kudin.
\newblock Bent functions and the completed maiorana-mcfarland class, 2026.

\bibitem[Sho97]{DBLP:journals/siamcomp/Shor97}
Peter~W. Shor.
\newblock Polynomial-time algorithms for prime factorization and discrete
  logarithms on a quantum computer.
\newblock {\em {SIAM} J. Comput.}, 26(5):1484--1509, 1997.

\bibitem[Sim97]{DBLP:journals/siamcomp/Simon97a}
Daniel~R. Simon.
\newblock On the power of quantum computation.
\newblock {\em {SIAM} J. Comput.}, 26(5):1474--1483, 1997.

\bibitem[YHLW21]{DBLP:journals/iandc/YeHLW21}
Zekun Ye, Yunqi Huang, Lvzhou Li, and Yuyi Wang.
\newblock Query complexity of generalized simon's problem.
\newblock {\em Inf. Comput.}, 281:104790, 2021.

\end{thebibliography}

\appendix

\section{Relation to Maiorana--McFarland and to Girish--Servedio}
\label{sec:gs_relation}
In this appendix, we provide some further background on the class of bent functions used in our lower bound, and relate our construction to that of Girish and Servedio (which relied on a subset of these functions).
\subsection{The Maiorana--McFarland Class}
\label{sec:mm_relation}
We represent a string in $\F_2^n$ as a concatenation of two strings, $(u,v)$ with $u,v\in\F_2^m$, each being $m$ length strings. The
\emph{Maiorana--McFarland class} $\mathcal{M}$ consists of the functions
\begin{equation*}
    f'(x)=\langle u,\sigma(v)\rangle+\tilde{h}(v), \qquad x=(u,v)\in \F_2^n
\end{equation*}
with $\sigma$ a permutation of $\F_2^m$ and $\tilde{h}\colon\F_2^m\to\F_2$
arbitrary. One can show that every such $f'$ is a bent boolean function. The definition of this class of functions dates back to
McFarland~\cite{McFarland73}, who stated it for difference sets, and was found
independently by Maiorana; we refer to the survey of Pasalic et al.\ for a more detailed history~\cite[Section~1.2]{PasalicPZK26}. Its \emph{completion}
$\mathcal{M}^{\#}$ consists of all functions that are equivalent to some Maiorana--McFarland function up to affine transformations.
As Dillon~\cite{dillon1974elementary} characterized, a bent function lies in
$\mathcal{M}^{\#}$ exactly when it is affine on every coset of some
$m$-dimensional subspace of $\F_2^n$.

Our functions are in this completion. By~\eqref{eq:def_f}, $f$ is affine on every
coset of the $m$-dimensional subspace $K$, and it is bent by
\cref{lem:fourier_identity}, so $f\in\mathcal{M}^{\#}$; similarly
$g\in\mathcal{M}^{\#}$, with $K^\perp$ in place of $K$. When
$K=\F_2^m\times\{0\}$ the correspondence is exact: then,
$K^\perp=\{0\}\times\F_2^m$, and we may index both families of cosets by
$\F_2^m$, with $C_v=\F_2^m\times\{v\}$ and $D_u=\{u\}\times\F_2^m$, and take
the representatives $\phi_v=(0,v)$ and $\psi_u=(u,0)$. A string $x=(u,v)$ belongs to $C_v$, so by~\eqref{eq:def_f}
\begin{equation*}
    f(u,v)=\langle(u,v),(\pi(v),0)\rangle+h(v)=\langle u,\pi(v)\rangle+h(v),
\end{equation*}
the function of $\mathcal{M}$ with $\sigma=\pi$ and $\tilde{h}=h$. The string
$x$ is also in $D_u$, and since $\langle\phi_{\pi^{-1}(u)},\psi_u\rangle=0$,
\cref{eq:tie,eq:def_g} give that
$g(u,v)=\langle v,\pi^{-1}(u)\rangle+h(\pi^{-1}(u))+s$: a function of
$\mathcal{M}$ with the two halves exchanged, plus $s$.

\subsection{The Girish--Servedio Construction}\label{app:GS_construction}
As mentioned earlier, the construction of Girish and Servedio~\cite{GirishS26} can
be viewed as a special case of ours, in which the cosets carry labels from
$\F_2^{m}$ and the permutation is a translation of those labels rather
than a uniformly random bijection. As a recap, Girish and Servedio build their hard instances from the functions of
$\mathcal{M}$ with $\sigma$ equal to the identity, $(u,v)\mapsto\langle
u,v\rangle+\tilde{h}(v)$, and then apply an invertible affine map to the
input, which keeps them in $\mathcal{M}^{\#}$.\footnote{Note that including this affine map is crucial: indeed, the
dual of $(u,v)\mapsto\langle u,v\rangle+\tilde{h}(v)$ is
$(u,v)\mapsto\langle u,v\rangle+\tilde{h}(u)$, which agrees with it at the
origin $(0,0)$, so without a shift (the parameter $a$ below) the \yes- and \no-instances could easily be told apart by
querying both functions on the all-zero input.} Their
distribution (\cite[Definition 2.2]{GirishS26}) is the
following: (1)~sample an invertible matrix $A\in\F_2^{n\times n}$ uniformly at random
and let $B=(A^{T})^{-1}$; for a matrix with $n$ columns, a subscript $1$ or $2$
denotes its top or its bottom half, a matrix with $m$ rows. (2)~Sample
$a\in\F_2^{n}$ uniformly and let $b=B^{T}v$, where $v$ is the vector
whose top half is $B_2a$ and whose bottom half is $B_1a$. (3)~Sample
$\tilde{h}\colon\F_2^{m}\to\F_2$ uniformly. The instance is
then the pair $(\chi_f,\chi_g)$ given by
\begin{align*}
    f(x)&=\langle A_1x,A_2x\rangle+\langle x,a\rangle+\tilde{h}(A_2x),\\
    g(y)&=\langle B_1y,B_2y\rangle+\langle y,b\rangle+\tilde{h}(B_1y+B_1a)
          +\langle B_1a,B_2a\rangle+s.
\end{align*}
The bit $s$ has the same meaning as in \cref{sec:construction}. (Girish and
Servedio write their \no-instances as $(\chi_f,-\chi_g)$, which is the pair
above with $s=1$.) It can be verified that by taking $K=\ker A_2$ (which implies $K^\perp=\ker B_1$), choosing the permutation $\pi(w)=w+B_1a$ and setting the random function as:
\begin{equation*}
    h(w)=\tilde{h}(w)
             +\langle\phi_w,A_1^{T}w+a+\psi_{\pi(w)}\rangle
\end{equation*}
in our construction, one recovers the hard instance used by Girish and Servedio. To summarize, the choice of $\pi$, the permutation, is the main (and crucial) addition in our construction, since the random subspace $K$ and the uniform random function $h$ are already indirectly used by Girish and Servedio.
\section{Upper Bounds for the Construction}
\label{sec:upper_bounds}

The construction of \cref{sec:construction} samples three objects: the subspace
$K$, the permutation $\pi$ and the bits $h$. In this appendix, we focus on how many queries are enough
to solve extremal Forrelation on our instances, first as defined, and then
when $\pi$ or $K$ is not drawn at random but is fixed and known to the
algorithm. %
\begin{itemize}
    \item Whatever the permutation, $O(2^{n/2})$ deterministic non-adaptive queries are enough to solve the problem. This is an adaptation of a known algorithm by Bravyi et al.~\cite{Bravyi2021}, who give a classical algorithm for the general Forrelation problem with query complexity $\tilde{O}(2^{n/2})$.
    \item With the permutation fixed (and so in particular, by the previous section, on the instances of Girish and Servedio), $\tilde{O}(2^{n/4})$ queries are enough. This shows the necessity of the permutation $\pi$ in order to obtain our improved lower bound.
    \item With $K$ fixed and known, (adaptive) $m+2$ queries settle the problem. This shows the necessity of the random subspace $K$ in order to prove the lower bound. (We once more emphasize that the inclusion of randomly sampling $K$ is not new, as it is indirectly used in the construction of Girish and Servedio.)
    \item Even with $K$ drawn at random, we show how an \emph{adaptive} algorithm can identify it with $\tilde{O}(2^{3n/8})$ queries and then apply the previous algorithm to solve the instance. Therefore the non-adaptive hypothesis in \cref{thm:main} is necessary.
\end{itemize}

\subsection{An Algorithm with \texorpdfstring{$O(2^{n/2})$}{O(2 to the n/2)}
Queries}
\label{sec:subspace_algorithm}

Bravyi et al.~\cite{Bravyi2021} solve extremal Forrelation on every instance
with a randomized non-adaptive algorithm making $O(2^{n/2})$ queries: at a high level, it
evaluates the sum defining Forrelation over two random affine subspaces of
dimension roughly $m$~\cite[Theorem~1]{Bravyi2021}. Inspired by this, we show a similar $O(2^{n/2})$-query algorithm solves our hard instances with high probability.

For a subspace $U$ of $\F_2^{n}$ define
\begin{equation}\label{eq:def_forr_U}
    \forr_U(\chi_f,\chi_g)=\frac{2^{n/2}}{|U|^{2}}
    \sum_{x\in U}\sum_{y\in U}\chi_f(x)\,\chi_g(y)\,(-1)^{\langle x,y\rangle},
\end{equation}
which is \cref{def:fourier_forrelation} with $\widehat{\chi_f}$ expanded and
both sums restricted to $U$, so that
$\forr_{\F_2^{n}}(\chi_f,\chi_g)=\forr(\chi_f,\chi_g)$.
Our algorithm fixes a subspace $U$ of dimension $m+k$ in advance, for a
constant $k$, and queries the functions on every $x\in U$. From the $2^{k}M$ answers, it then computes $\forr_U(\chi_f,\chi_g)$ and reports a \yes-instance if, and only if, $\forr_U(\chi_f,\chi_g) > 0$. Here, the choice of $U$ is arbitrary, and the probability of success can be boosted to any arbitrarily small constant by increasing $k$.

\begin{lemma}\label{lem:restricted_forr}
    For every instance of \cref{sec:construction} and every subspace $U$ of
    $\F_2^{n}$,
    \begin{equation*}
        \forr_U(\chi_f,\chi_g)=(-1)^{s}\Bigg(1+\frac{1}{|U|}
        \sum_{\substack{w\in U^{\perp}\\w\neq0^n}}\ \sum_{y\in U}
        \chi_g(y)\,\chi_g(y+w)\Bigg).
    \end{equation*}
\end{lemma}
\begin{proof}
    By \cref{lem:character_sum}, the sum
    $\sum_{w\in U^{\perp}}(-1)^{\langle x,w\rangle}$ is $|U^{\perp}|$ when
    $x\in(U^{\perp})^{\perp}=U$ and $0$ if $x\notin U$. Fixing $y$, we thus have, expanding the definition of the Fourier transform,
    \begin{align*}
        |U|\sum_{w\in U^{\perp}}\widehat{\chi_f}(y+w)
        &=\frac{|U|}{2^{n}}\sum_{x\in\F_2^{n}}
          \chi_f(x)\,(-1)^{\langle x,y\rangle}
          \sum_{w\in U^{\perp}}(-1)^{\langle x,w\rangle}\\
        &=\frac{|U|\,|U^{\perp}|}{2^{n}}\sum_{x\in U}
          \chi_f(x)\,(-1)^{\langle x,y\rangle}
         =\sum_{x\in U}\chi_f(x)\,(-1)^{\langle x,y\rangle},
    \end{align*}
    since $\dim(U^{\perp})=n-\dim(U)$ and so $|U||U^{\perp}|=2^{n}$.
    Now, swapping the sums in~\eqref{eq:def_forr_U}, and using the above along with the fact that
    $\widehat{\chi_f}=2^{-n/2}(-1)^{s}\chi_g$ (by
    \cref{lem:fourier_identity}), we get
    \begin{align*}
        \forr_U(\chi_f,\chi_g)
        &=\frac{2^{n/2}}{|U|^{2}}\sum_{y\in U}\chi_g(y)
          \sum_{x\in U}\chi_f(x)(-1)^{\langle x,y\rangle}\\
        &=\frac{(-1)^{s}}{|U|}\sum_{y\in U}\sum_{w\in U^{\perp}}
          \chi_g(y)\,\chi_g(y+w)\\
        &=(-1)^{s}\Bigg(1+\underbrace{\frac{1}{|U|}
          \sum_{\substack{w\in U^{\perp}\\w\neq0^n}}\sum_{y\in U}
          \chi_g(y)\,\chi_g(y+w)}_{\text{the ``correction'' term}}\Bigg),
    \end{align*}
    the last step because $\chi_g(y)^{2}=1$, so the $w=0^n$ term contributes
    $|U|^{-1}\sum_{y\in U}\chi_g(y)^{2}=1$.
\end{proof}
Writing $R$ for this ``correction term,'' \cref{lem:restricted_forr} can be rewritten as
$\forr_U(\chi_f,\chi_g)=(-1)^{s}(1+R)$. This means the algorithm is correct
whenever $1+R>0$. Our goal for the rest of the subsection is then to show that this
happens with good probability over the draws of $K$ and $h$, whatever the
permutation.

The next claim shows that if $U^\perp$ has trivial intersection with both $K$ and $K^\perp$, then $R$ is small.
\begin{claim}
We have, with the notation above, and for any fixed $K$ and $\pi$:
\begin{enumerate} 
    \item if $U^\perp\cap K^\perp=\{0^n\}$, then $\mathbb{E}[R]=0$.
    \item if $U^\perp\cap K^\perp=\{0^n\}$ and $U^\perp\cap K=\{0^n\}$, then $\Pr[|R|\geq 1]\leq 1/2^{k-1}$.
    \item The above two conditions (that is, $U^\perp$ has trivial intersection with both $K$ and $K^\perp$) simultaneously hold with probability at least $1-1/2^{k-1}$.
\end{enumerate}
\end{claim}
Note that $1+R$ is positive whenever $|R|<1$, and then $\forr_U(\chi_f,\chi_g)$ has the same sign as $(-1)^s$, leading to the algorithm returning the correct answer. Assuming the claim holds, we have simultaneously $U^\perp\cap K^\perp=\{0^n\}$ and $U^\perp\cap K=\{0^n\}$ with probability at least $1-1/2^{k-1}$, in which case the second item guarantees $|R|< 1$ with probability at least $1-1/2^{k-1}$: by a union bound, the algorithm errs with probability at most $1/2^{k-2}$, making $2^kM=O(2^{n/2})$ queries  (as $k$ is constant).\smallskip

We now prove the three items of the claim, starting with the first.
\begin{proof}[Proof of first item]
    Each term of the correction term $R$ is of the form $\chi_g(y)\chi_g(y+w)$. Assume $y\in D_j$, and let $i(y)=\pi^{-1}(j)$ be the index of the coset of $K$ that is paired with the coset of $K^\perp$ containing $y$ under $\pi$. Then, $i(y)=i(y')$ exactly when $y$ and $y'$ are in the same coset of $K^\perp$. Recall the definition of our function $g$:
    \begin{equation*}
    g(y)=\underbrace{\langle\phi_{i(y)},y\rangle
        +\langle\phi_{i(y)},\psi_{j}\rangle+s}_{\eta(y)}+h(i(y)).
\end{equation*}
Then, $g(y)+g(y+w)=\eta(y)+\eta(y+w)+h(i(y))+h(i(y+w))$. Recall that, conditioned on $K$ and $\pi$, $\eta(\cdot)$ is fixed, and thus the expectation of $R$ only depends on the randomness of $h$. Assuming $w\not\in K^\perp$, then $i(y)\neq i(y+w)$. Under a random $h$, this means $h(i(y))+h(i(y+w))$ is a uniform random bit, and thus $g(y)g(y+w)$ is also a uniform random bit. Summing over all $y$, assuming for every $w$, $w\not\in K^\perp$, this gives that $R$ has mean zero, and thus $\mathbb{E}[R]=0$, as required. 
\end{proof}
\noindent We then turn to the second item:
\begin{proof}[Proof of second item]
    We will use Chebyshev's inequality to bound this probability. To do so, we first need to compute $\mathbb{E}[R^2]$:
    \begin{equation*}
    \mathbb{E}[R^{2}]=\frac{1}{|U|^{2}}\sum_{w,w'\in U^\perp\setminus\{0^n\}}\sum_{y,y'\in U}
    (-1)^{\eta(y)+\eta(y+w)+\eta(y')+\eta(y'+w')}\,
    \mathbb{E}\big[(-1)^{h(i(y))+h(i(y+w))+h(i(y'))+h(i(y'+w'))}\big].
\end{equation*}
Since the terms outside the expectation are fixed under $K$ and $\pi$, we will only focus on the terms inside the expectation, which are given by the random function $h$. Using the assumption $U^\perp\cap K^\perp=\{0^n\}$, we have that $w,w'\not\in K^\perp$, and hence $i(y)\neq i(y+w)$ and $i(y')\neq i(y'+w')$. Since $h$ is random, the expectation of this term will be zero, unless either $i(y)=i(y')$ and $i(y+w)=i(y'+w')$ or $i(y)=i(y'+w')$ and $i(y+w)=i(y')$, in which case the expectation equals one. We now count the number of ways these last two cases can occur.

First, notice in either case, we then have $w+w'\in K^\perp$, by subspace closure properties. Since $w+w'\in U^\perp$, our assumption $U^\perp\cap K^\perp=\{0^n\}$ implies that $w=w'$. This then also implies that $y'$ is either in the coset $y+K^\perp$ or $y+w+K^\perp$. Now, our second assumption $U^\perp\cap K=\{0^n\}$ implies, taking orthogonal complements, that the span of $U$ and $K^\perp$ equals $\mathbb{F}_2^n$, and thus $\dim(U\cap K^\perp)\leq k$. Hence, the intersection of $U$ with the coset $y+K^\perp$ is of size at most $2^k$. The same goes for the coset $y+w+K^\perp$. We note there are $2^{m-k}-1$ valid choices of $w$, and that the product of the $\eta(\cdot)$ terms outside the expectation is at most $1$. Finally, there are $|U|$ choices of $y$, and for each choice of $y$, we know $y'$ must be in $U$, and also either the coset $y+K^\perp$ or $y+w+K^\perp$. Each of these cosets has intersection at most $2^k$ with $U$, and so there are $2^{k+1}$ choices of $y'$. Putting this all together to bound the second moment:
\begin{equation*}
    \mathbb{E}[R^{2}]\leq\frac{1}{|U|^{2}}\,(2^{m-k}-1)\,|U|\,2^{k+1}
    \leq\frac{2^{m-k}\,2^{k+1}}{2^{m+k}}=1/2^{k-1}.
\end{equation*}
To conclude, since $\mathbb{E}[R]=0$ (by the first item), $\operatorname{Var}[R]=\mathbb{E}[R^2]$, and Chebyshev's inequality gives
\begin{equation*}
    \Pr[|R|\geq 1]\leq 1/2^{k-1},
\end{equation*}
as required.
\end{proof}
\noindent We finally conclude with the proof of the third item:
\begin{proof}[Proof of third item]
    By \cref{lem:subspace_count}, any single nonzero vector is included in the uniformly random $m$-dimensional subspace $K^\perp$ with probability at most $2^{-m}$. Since $U^\perp$ has dimension $m-k$, taking a union bound over its $2^{m-k}-1$ nonzero elements, we get that
    \begin{equation*}
        \Pr[U^\perp\cap K^\perp\neq \{0^n\}]\leq 2^{-m}(2^{m-k}-1)\leq 2^{-k}\,.
    \end{equation*}
    By the same argument, we get that 
    \begin{equation*}
        \Pr[U^\perp\cap K \neq \{0^n\}]\leq 2^{-k},
    \end{equation*}
    and the conclusion follows by a union bound over the two events.
\end{proof}

 \subsection{The Role of the Permutation}
\label{sec:permutation}

Recall that to solve an instance of our construction, by~\cref{lem:coset_rep_independence} it is enough to know its underlying subspace $K$
and permutation $\pi$ it was built from, since querying any two paired cosets reveals~$s$. Intuitively, having the random permutation in addition to the random subspace $K$ should make it harder for the adversary by ``adding more randomness'', compared to the construction of Girish and Servedio in which $\pi$ is absent. In this subsection, we quantify this intuition.

To do this, we consider a slightly more general setting, where the pair $(K,\pi)$ belongs to a set $\mathcal{L}$ of candidates known to the algorithm. The Forrelation instance is built from a candidate $(K,\pi)\in\mathcal{L}$ (which can be drawn from \emph{any} probability distribution over $\mathcal{L}$): by suitably choosing $\mathcal{L}$, this generalizes both our construction and that of~\cite{GirishS26}. 
We present an algorithm which, in this setting, identifies the pair $(K,\pi)$ using
$O(2^{n/4}\log|\mathcal{L}|)$ queries (\cref{lem:candidate_algorithm}).  To see what this implies: in the construction of~\cite{GirishS26}, since the hardness instance is constructed by uniformly at random sampling $a\in\mathbb{F}_2^n$ and the full rank matrix $A$, there are at most $2^{O(n^2)}$ different candidates. Thus, $\log|\mathcal{L}|=O(n^2)$, which gives an upper bound of $\tilde{O}(2^{n/4})$. More generally, this bound on $|\mathcal{L}|$ holds if the chosen permutation $\pi$ is fixed, by \cref{lem:subspace_count}. \smallskip

The algorithm draws $r=O(\log|\mathcal{L}|)$ \emph{blocks} of $b=2^{n/4}$ points uniformly at random, and (non-adaptively) queries $f$ and $g$ at each of these $q=2rb$ points.\footnote{To get an intuition on the setting of parameters: since $b^{2}=M$, by a standard birthday paradox argument, with constant probability in each block $\Omega(1)$ cosets of $K$ will contain a collision.} For each candidate $(K',\pi')\in\mathcal{L}$, for each pair of queries $x_1$ and $x_2$ belonging to the same coset $C_i'$ of $K'$, the algorithm checks that the identity $f(x_1)+f(x_2)=\langle x_1+x_2,\psi'_{\pi'(i)}\rangle$ holds, where $C_i'$ and $\psi_j'$ are the cosets of $K'$ and the representatives of the cosets of $K'^{\perp}$ that \cref{sec:construction} defines for $K'$. The algorithm discards every candidate the transcript contradicts. Then, for each of the remaining candidates in $\mathcal{L}$, with constant probability there will be two queried points $x,y$ whose corresponding cosets have been paired by a surviving candidate. With any such queried pair, it outputs $f(x)+g(y)+\langle x,y\rangle$ as the answer; if no candidate survives, or no kept candidate pairs the cosets of two queried points, it outputs a uniform bit.

Observe that, for two candidates sharing the same subspace $K$ (only differing in the permutation), the corresponding matchings between $K$ and $K^\perp$ can differ in as few as two edges between the two candidates, which makes it unlikely the algorithm can distinguish between these two candidates. We will say the set of candidates $\mathcal{L}$ is \emph{separated} if for any two candidates constructed from the same subspace $K$, the choice of permutations will match at least half of the cosets of $K$ to different cosets of $K^\perp$. We then show the following guarantee on the above algorithm:

\begin{lemma}\label{lem:candidate_algorithm}
    Let $\mathcal{L}$ be separated with
    $\log|\mathcal{L}|\leq 2^{n/4}$. The algorithm above, run with
    $r=O(\log |\mathcal{L}|)$ blocks, makes
    $O(2^{n/4}\log|\mathcal{L}|)$ queries and outputs $s$ with probability at
    least $2/3$, for every candidate in $\mathcal{L}$, every $s$, and every choice of random function $h$.
\end{lemma}
\begin{proof}[Proof of \cref{lem:candidate_algorithm}]
    We will make use of the Paley--Zygmund inequality:
    \begin{lemma}[Paley--Zygmund]
        If $Z\geq 0$ is a random variable with $\mathbb{E}[Z^2] < \infty$, then, for every $\theta\in [0,1]$,
        \begin{equation*}
            \Pr[Z>\theta\mathbb{E}[Z]]\geq (1-\theta)^2\frac{\mathbb{E}[Z]^2}{\mathbb{E}[Z^2]}.
        \end{equation*}
    \end{lemma}

    We now consider an incorrect candidate $(K',\pi')$, and compare it to the correct candidate $(K,\pi)$. We want to bound the probability that the incorrect candidate still is somehow consistent with the generated transcript. There are two cases for this:
    \begin{description}
        \item[Case 1, $K'\neq K$:] Since this candidate is incorrect, in one of the $r$ blocks it must (with high probability) have two queried points $x_1$ and $x_2$ in the same coset (call it $C'_i$), which under $K$ are in different cosets. This means, due to coset closure properties, that we can find $x_1+x_2\in K'\setminus K$. Then, under the wrong candidate, we will have
    \begin{equation*}
        f(x_1)+f(x_2)=\langle x_1+x_2,\psi_{\pi'(i)}\rangle, 
    \end{equation*}
    with the random function canceled out. On the other hand under the true candidate, $x_1\in C_{i_1}$ and $x_2\in C_{i_2}$, and hence
    \begin{equation*}
        f(x_1)+f(x_2)=\langle x_1,\psi_{\pi(i_1)}\rangle+\langle x_2,\psi_{\pi(i_2)}\rangle+h(i_1)+h(i_2),
    \end{equation*}
    which is a uniformly random bit. As a result, the wrong and true candidates, on this specific pair, disagree with probability $1/2$. 
    
    We will now proceed through every block, and try to find a pair of queries such that their sum is in $K'\setminus K$, and their true cosets in $K$ have not been examined in prior blocks (that is, no coset is reused between blocks). (This second condition is for the analysis, so that each test is independent from the prior tests.) Before examining any block, at most $2r$ cosets have been previously examined, which by assumption, is less than $O(M)$. For two random points, $x_1$ and $x_2$, we have that 
    \begin{equation*}
        \Pr[x_1+x_2\in K'\setminus K]=\Theta(1/M), 
    \end{equation*}
    since $|K'\setminus K|$ is between $M/2$ and $M$, as they have to differ on at least one basis vector. Removing the $2r$ previously queried cosets only loses at most a constant factor in this calculation. Let $Z$ be the indicator of the event that in a block, there is a pair satisfying the above two conditions. Then, taking a probability over all $b=O(\sqrt{M})$ queries in this block, we get that $Z$ occurs with probability $\Theta(b^2/M)=\Theta(1)$, i.e.,  $\mathbb{E}[Z]=\Theta(1)$. Similarly, one can show that  $\mathbb{E}[Z^2]=O(1)$. Thus,  applying the Paley--Zygmund inequality with $\theta=0$, we get 
    \begin{equation*}
        \Pr[Z>0]\geq c\,,
    \end{equation*}
    for some constant $c>0$. Let $N_t$ be the number of pairs which satisfy the above criterion in the first $1\leq t\leq r$ blocks. So far, the wrong candidate is consistent with the transcript with probability $2^{-N_t}$, and so, on the next block, we have that
    \begin{equation*}
        \mathbb{E}[2^{-N_{t+1}} \mid \text{prior blocks}] \leq c\cdot\frac{1}{2}2^{-N_t}+(1-c)2^{-N_t}=\left(1-\frac{c}{2}\right)2^{-N_t}.
    \end{equation*}
    Thus, each block shrinks the expected survival probability of the wrong candidate by a factor of $1-c/2$. Iterating this over the $r$ blocks, we get that
    \begin{equation*}
        \Pr[\text{wrong candidate $(K',\pi')$ survives}]\leq \left(1-\frac{c}{2}\right)^r.
    \end{equation*}

    \item[Case 2, $K=K'$ but $\pi'\neq \pi$:] In this case, we condition on finding some pair of queries, $x_1$ and $x_2$ falling in the same coset $C_i$, which implies $x_1+x_2\in K\setminus\{0^n\}$. This occurs with probability $O(1/M)$, for a fixed pair of queries. Now, the true permutation $\pi$ pairs this coset with $D_{\pi(i)}$, whereas $\pi'$ pairs this with $D_{\pi'(i)}$. Since we assume $\mathcal{L}$ is well separated, 
    $
        \Pr_i[\pi(i)\neq\pi'(i)]\geq\frac{1}{2}.
    $
    Then, assuming this happens, the incorrect candidate gives, setting $u=x_1+x_2$,
    \begin{equation*}
        f(x_1)+f(x_2)=\langle u,\psi_{\pi'(i)}\rangle,
    \end{equation*}
    whereas the correct candidate gives
    \begin{equation*}
        f(x_1)+f(x_2)=\langle u,\psi_{\pi(i)}\rangle
    \end{equation*}
    and so they disagree whenever $\langle u,\psi_{\pi'(i)}\rangle\neq \langle u,\psi_{\pi(i)}\rangle$. Since the permutations map to different cosets, there must be some $a\in K$ such that 
    \begin{equation*}
        \langle a,\psi_{\pi'(i)}\rangle\neq \langle a,\psi_{\pi(i)}\rangle.
    \end{equation*}
    Now, pair each $u\in K$ with $u+a\in K$: if the incorrect and correct candidates agree at $u$, they must disagree at $u+a$, and vice versa. Thus, half of the elements of $K$ will ``catch'' this incorrect permutation, i.e., $u\in K$ catches this incorrect candidate with probability at least $1/2$. Putting these two probabilities together, this means that a random pair catches the incorrect candidate with probability at least $\Theta(1/M)$. Using the same Paley--Zygmund argument as in Case 1, analyzing across all $r$ blocks, this then means
    \begin{equation*}
        \Pr[\text{wrong candidate $(K',\pi')$ survives}]\leq \left(1-\frac{c'}{2}\right)^r\,,
    \end{equation*}
    for some $c'>0$.
    \end{description}
 Putting the two cases together (and redefining $c$ to be the minimum of $c,c'>0$), by a union bound, we have at most $O(|\mathcal{L}|)$ incorrect candidates, each of which will be detected with probability at least $1-\left(1-\frac{c}{2}\right)^r$. Hence, 
    \begin{equation*}
        \Pr[\text{some wrong candidate survives}]\leq O(|\mathcal{L}|)\left(1-\frac{c}{2}\right)^r.
    \end{equation*}
    Finally, assuming we now have identified the true candidate $(K, \pi)$, the algorithm still needs to find a pair of cosets $C_i$ and $D_j$ that have been paired by the queries, so they can be queried and the secret bit $s$ recovered. Take a pair of queries $x$ and $y$, and fix $x$. Then, the probability that $y\in D_{\pi(i)}$, given that $x\in C_i$, is $\Theta(1/M)$. Using yet another Paley--Zygmund argument, we can show that
    \begin{equation*}
        \Pr[\text{no block contains such a pair}]\leq \left(1-\frac{c}{2}\right)^r,
    \end{equation*}
    where again, we redefine the $c>0$ used in all 3 cases. Putting all three together then, we get that
    \begin{equation*}
        \Pr[\text{the algorithm is incorrect}]\leq O(|\mathcal{L}|)\left(1-\frac{c}{2}\right)^r,
    \end{equation*}
    which can be made arbitrarily small by adjusting the constant factor in the definition of $r$.
\end{proof}
\subsection{The Role of the Subspace}
\label{sec:subspace}
In this subsection we show that if the subspace $K$ is fixed in advance and known to the algorithm, so that only
$\pi$ and $h$ are sampled, then there is an adaptive $m+2$ query algorithm that solves the Forrelation instance. While the lower bound of \cref{thm:main} does not apply here, since that is only a lower bound on non-adaptive algorithms, the $\tilde{\Omega}(2^{n/4})$ lower bound of Girish and Servedio does: thus, this shows the task becomes exponentially easier if $K$ is known.

\begin{lemma}\label{lem:known_subspace}
    Fix a subspace $K$ of dimension $m$, a permutation $\pi$, a function $h$
    and a bit $s$, and let $f$ and $g$ be the instance as defined in \cref{sec:construction}. There is an algorithm which, given knowledge of $K$, makes
    $m+1$ non-adaptive queries, then adaptively queries one additional point, and
    outputs $s$ with probability one.
\end{lemma}
\begin{proof}
    Fixing a basis $(u_1,\dots,u_m)$ of $K$, the algorithm queries the $m+1$ points
    $0,u_1,\dots,u_m$, all of which belong to the coset $C_i=K$. Since $f$ is
    $\langle\cdot,\psi_{\pi(i)}\rangle+h(i)$ on $C_i$,
    \begin{equation*}
        f(0)+f(u_\ell)=\langle u_\ell,\psi_{\pi(i)}\rangle
                      =\langle u_\ell,D_{\pi(i)}\rangle
        \qquad\text{for }1\leq\ell\leq m,
    \end{equation*}
    in the notation of \cref{lem:coset_pairing}. That lemma, applied to this
    basis $u_1,\dots,u_m$ of $K$, guarantees that no two cosets of $K^{\perp}$ share the list of values
    $\langle u_1,D\rangle,\dots,\langle u_m,D\rangle$. The algorithm knows
    $K^{\perp}$ and can compute the list of every coset, so the $m$ bits above
    uniquely identify $D_{\pi(i)}$. Once $D_{\pi(i)}$ is known, the algorithm queries an arbitrary $y\in D_{\pi(i)}$. By
    \cref{lem:coset_rep_independence} applied to $x=0\in C_i$ and this $y$,
    $g(y)=f(0)+\langle 0,y\rangle+s$, i.e., $s=f(0)+g(y)$, which the algorithm can now compute.
\end{proof}

\subsection{An Adaptive Algorithm with
\texorpdfstring{$\widetilde{O}(2^{3n/8})$}{O tilde of 2 to the 3n/8} Queries}
\label{sec:adaptive_algorithm}
In this section, we outline an adaptive algorithm that learns the subspace $K$
of our construction using $\widetilde{O}(2^{3n/8})$ queries, with high
probability over the choice of $K$, $\pi$ and $h$. Together with
\cref{lem:known_subspace}, which solves the instance with $m+2$ further
queries once $K$ is known, this distinguishes our \yes- and \no-instances
with $\widetilde{O}(2^{3n/8})$ queries in all, so the non-adaptive hypothesis in
\cref{thm:main} is necessary. The argument is specific to our construction,
so it does not rule out the conjecture of Girish and
Servedio~\cite{GirishS26} that extremal Forrelation needs
$\widetilde{\Omega}(2^{n/2})$ queries in the adaptive regime. Here we only
sketch the algorithm and omit its analysis.

We call a two-dimensional subspace of $\F_2^n$ a \emph{plane}. Over $\F_2^n$, a
plane only contains 4 points, and so does each of its cosets. By \eqref{eq:def_f} the function $f$ is affine on each coset of $K$, and an affine function sums to zero over any coset of a plane. So when a plane $P$ lies inside $K$, the 4 points of a coset $x+P$ all lie in one coset of $K$, and the 4 values of $f$ sum to zero. Conversely, it can be shown that when $P$ does not lie inside $K$ this sum vanishes with probability bounded away from
$1$, over the choice of $\pi$ and $h$. Our tester thus essentially works by querying 4 points, and testing a plane for membership in $K$ by checking the corresponding sum: it will never reject a plane that lies inside $K$, and reject any other with constant probability. Then, once enough disjoint planes (in $K$) are found, this gives sufficiently many linearly independent vectors to reconstruct the subspace $K$.

To find one plane of $K$, fix a subspace $U$ of dimension $m+2$: we call every plane that is a subset of $U$ a \emph{candidate plane}. Since $\dim(U\cap K)\geq2$, $U$ must contain a plane of $K$, and it is enough to run the above test on the candidate planes of $U$ and keep any one that is never rejected: by the above, this plane (with high probability) is contained in $K$. To do this, query a random subset $X\subseteq U$, taking each point independently with some suitable probability $\delta$. We require that for a given plane, all 4 points lie in $X$, which happens with probability $\delta^{4}$. Each plane has a total of $|U|/4$ cosets in $U$, so setting $\delta\approx |U|^{-1/4}$ gives each plane at least one testable coset in expectation. However, this only tests a candidate plane once: to increase the probability that a ``bad'' plane gets rejected, we need to test more than one coset per plane. It can be shown that we need about $\Theta(n)$ cosets of one candidate plane to be correct with high probability. To capture enough points to allow for this, it can be shown that raising $\delta$ by a factor of $n^{1/4}$ gives with high probability, $\Theta(n)$ different cosets of \emph{every} candidate plane in $U$. At the end of this step, with high probability, all surviving candidate planes will be contained in $K$, and the set of points queried, $X$, has size
$\widetilde{O}(|U|^{3/4})=\widetilde{O}(2^{3n/8})$.

Once we have a plane, we also know a nonzero $p\in K$, and the remaining spanning vectors of $K$ are easier to guess, because the algorithm can query in pairs: for each sampled point
$x$ we query both $x$ and $x+p$. When $x$ and $x+z$ are both sampled, their two pairs make up a coset of the plane spanned by $p$ and $z$, so a test now costs 2 sampled points rather than 4. That is, instead of testing candidate planes, we are testing candidate lines. Then $\delta$ of order $|U|^{-1/2}$ already gives each candidate $z$ a testable coset, and along with a further $\operatorname{poly}(n)$ factor to increase the probability of rejection as before, this leads to a sample of size $\widetilde{O}(|U|^{1/2})=\widetilde{O}(2^{n/4})$ to query. Finally, one can show that $O(n)$ such rounds, each on a fresh subspace, are enough to recover all of $K$. The query complexity is then dominated by the $\widetilde{O}(2^{3n/8})$ queries in the first part of the test (to obtain a nonzero $p\in K$).

The algorithm is indeed adaptive, as it uses an answer to choose a query in two places: when it picks $p$ from the plane it has found, and (once it has recovered $K$) at the last query of \cref{lem:known_subspace}.

\section{In passing: An optimal lower bound for the Generalized Simon's Problem}\label{sec:gsp_proof}

As discussed in the introduction, (a simpler version of the) ideas similar to those underlying our main result on extremal Forrelation can be leveraged to obtain a tight bound on another problem, the Generalized Simon's problem, establishing~\cref{thm:generalized:simons}.  

We first formally define the problem considered. Let $p$ throughout be a prime,  and choose uniformly at random a subgroup $S$ such that $\dim(S)=k$, for some constant $k$. We then consider a function $f\colon\F_p^n\rightarrow X$, for some sufficiently large universe $X$,\footnote{Ye et al.~\cite{DBLP:journals/iandc/YeHLW21} define this on the group $\mathbb{Z}_p^n$, but this turns out to be equivalent.} with the promise that for two vectors, if $x+y\in S$, then $f(x)=f(y)$ (where addition is taken as coordinate-wise addition modulo $p$). In this sense, this function $f$ is a generalization of Simon's problem, where instead of being promised the function is $2$-to-$1$, we are promised it is $p^k$-to-$1$. An adversary is given query access to this function $f$, and their goal is to recover the subgroup $S$ with high probability. We denote an instance of this problem as $\gsp(n,k,p)$. %

\paragraph{(Some) group theory} We also briefly recall some group-theoretic notation. For a subgroup $H$, and a fixed vector $g$, we let $g+H$ be the set $\{g+h:h
\in H\}$. We let $\langle a_1,\cdots ,a_k\rangle$ denote the set of vectors generated by $a_1,\dots,a_k$. In other words, it is the set of vectors spanned by $\{a_1,\dots,a_k\}$. We call the $a_i$ the generators. A set of generators can be seen to partition the full space into cosets as follows. Let $A=\{a_i\}$ be a set of generators for a subgroup $H$. Then, $A$ partitions the space $\F_p^n$ into $p^n/|H|$ cosets: to obtain the cosets, we sequentially choose any fixed vector $g_i$ (not already in a coset) and assign to its coset all the vectors in the set $g_i+H$, repeating until every vector in the space falls in a coset. We then call each $g_i$ a \emph{coset representative}.

We also use another common result in balls and bins analysis regarding the number of collisions created. Here, if we denote the bin where each ball ends up as a sequence of variables $x_1,\dots,x_\ell$, a collision is simply a pair $(x_i,x_j)$ where $x_i=x_j$ and $i\neq j$. 
\begin{lemma}\label{lemma:collision_count_bound}
    Let $C_\ell$ be the number of collisions after $\ell$ balls are tossed independently into $N$ bins according to the probability distribution $\mathcal{D}$. Let $k$ be any integer. Then
    \begin{equation*}
        \Pr[C_\ell\geq k]\leq \frac{{\ell\choose 2}\|\mathcal{D}\|_2^2}{k}.
    \end{equation*}
\end{lemma}
\noindent This is a standard result: we include the proof for completeness.
\begin{proof}
   Let $\mathcal{D}_j$ be the probability that any individual toss lands in bin $j$. Then, for any two tosses, the probability that they land in the same bucket and cause a collision is given by $\sum_j\mathcal{D}_j^2=\|\mathcal{D}\|_2^2$. By linearity of expectation, the expectation of $C_\ell$ is then
    \begin{align*}
        \mathbb{E}[C_\ell]={\ell\choose 2}\|\mathcal{D}\|_2^2
    \end{align*}
    By Markov's inequality, we get, for any $k > 0$,
    \begin{align*}
        \Pr[C_\ell\geq k]&\leq \frac{\mathbb{E}[C_\ell]}{k}=\frac{{\ell\choose 2}\|\mathcal{D}\|_2^2}{k}. \qedhere
    \end{align*}
\end{proof}\bigskip

We now leverage some of the ideas used for our main result on Forrelation, specifically the subspace counting, to prove a tight lower bound for adaptive algorithms for the Generalized Simon's Problem. Namely, we will prove the following.
\begin{theorem}\label{thm:tight_generalized_simons_bound}
    Given a problem instance $\gsp(n,k,p)$, no adaptive adversary can determine the subgroup $S$ without using at least $\ell=\Omega(\max\{k,\sqrt{kp^{n-k}}\})$ queries. Hence, determining $S$ requires $\Theta(\max\{k,\sqrt{kp^{n-k}}\})$ queries.
\end{theorem}
\begin{proof}
    Ye et al.\ already prove the necessity of the $k$ term in the first part of the max~\cite{DBLP:journals/iandc/YeHLW21}, so we only focus on proving the necessity of the $\sqrt{kp^{n-k}}$ in the second term. We assume the algorithm is adaptively querying and randomized, and results for deterministic algorithms follow accordingly. Assume some randomized algorithm makes at most $\ell$ queries, and outputs the correct subgroup $S$ with probability at least $\frac{2}{3}$. The way the function $f$ is chosen is as follows:
    \begin{enumerate}
        \item Choose uniformly at random a subgroup $S$ of dimension $k$
        \item Choose uniformly at random a $p^k$-to-one mapping compatible with $S$
    \end{enumerate}
    We show that no adaptive algorithm can identify the subgroup $S$ with probability greater than $2/3$ without making at least $\Omega(\sqrt{kp^{n-k}})$ queries.

    First we note that the subgroup $S$ partitions $\F_p^n$ into $p^{n-k}$ cosets, which we may interchangeably refer to as \emph{bins} or \emph{buckets} going forward. We say a collision occurs between queries $i$ and $j$ whenever these two queries are assigned to the same bucket. For the purpose of the Generalized Simon's Problem, notice that whenever a collision occurs between $x_i$ and $x_j$, the adaptive algorithm learns that $x_i+x_j\in S$ (alternatively, one can reframe a collision as there being two queries appearing in the sequence such that $x_i+x_j\in S$). 
    
    To get some intuition, notice now that the task of uncovering the subgroup $S$ is akin to finding $k$ collisions (not collisions in one bucket). This is because of the following observation: Naively, each collision would give the algorithm information on what vectors are in the subgroup $S$, but nonadaptive queries may produce the same vector multiple times, thus not learning anything new. (For example, if an algorithm's queries result in finding out that $f(x)=f(y)$ and $f(v)=f(w)$, but $x+y=v+w$, then nothing was gained by the algorithm, even though a collision occurred.)
    However, an adaptive algorithm can proceed as follows. Say on the first collision, they learn that $a_1\in S$. Then, they can group all strings into cosets $x+\langle a_1\rangle$, where $\langle a_1\rangle$ is the subgroup generated by $a_1$. Now, having done this, the adaptive algorithm never needs to query multiple vectors in a coset, since knowing what one evaluates to will give the value of the others. Thus, they can continue the querying process, making sure only to query the coset representatives induced by $a_1$. Doing this ensures an adaptive algorithm will never produce a collision that lets them learn $a_1\in S$ again. Thus, they can continue making queries, where they only query one in each coset, until they find a second collision, and this collision is guaranteed to inform them of a new linearly independent vector $a_2\in S$. Then, they can do the same partitioning with the vectors $a_1, a_2$ to obtain the coset representatives induced by $\langle a_1,a_2\rangle$, and ensure they only query the coset representatives from now going forward. This ensures the next collision informs them of another unique linearly independent vector $a_3$. This can be continued until the $k$-th collision, at which point the adaptive algorithm will know of all $k$ unique spanning vectors in $S$, and hence can identify the subgroup uniquely.

    With the above intuition, we proceed onto the formal proof. Consider a sequence of $\ell$ adaptive queries, $x_1,\dots,x_\ell$. We proceed to find the probability distributions $\mathcal{P}_1,\dots,\mathcal{P}_\ell$ for which the queries are assigned, conditioned over the randomness of the choice of $S$, assignment of the $p^k$-to-one mapping, and the values of the previous queries. Here $\mathcal{P}_i$ is the probability distribution for where the $i$-th query lands. For the first query $x_1$, clearly $\mathcal{P}_1$ is uniform over all the $p^{n-k}$ buckets. For the second query, conditional on the aforementioned events, $x_2$ is slightly negatively correlated with the bin that contains $x_1$. This is because this bin already contains one vector, whereas the others are empty, so conditioned on all the previous events, there is a slightly higher probability of $x_2$ being assigned to an empty bin rather than to $x_1$'s bin. However, we can assume $\mathcal{P}_2$ is uniform over all bins. This is because the probability that $x_2$ causes a collision with its true distribution is upper bounded by that if it were assigned uniformly. We can continue using the same reasoning for all $\ell$ queries, and assume all $\mathcal{P}_i$ are uniform over the bins.

    Taking $\mathcal{P}_i$ as the distribution, by \cref{lemma:collision_count_bound}, we have
    \begin{align*}
        \Pr[C\geq k]&\leq \frac \mu k
        =\frac{{\ell\choose 2}}{kp^{n-k}}
        =\frac{O(\ell^2)}{kp^{n-k}}.
    \end{align*}
    Hence, for this to have probability at least $2/3$, we require that $\ell$ is at least $\Omega(\sqrt{kp^{n-k}})$.

    Finally, we claim that if the query algorithm does not find at least $k$ collisions in $\ell$ queries, they cannot guess the hidden subgroup with high probability. Let $\mathcal{H}$ be the set of subgroups with dimension $k$, that also contain the vectors known to be in $S$ from the collisions in the query sequence $\ell$. By assumption, there are at most $k-1$ of these, and hence the size of $\mathcal{H}$ is bounded by the number of $k$-dimensional subgroups that also contain the given (at most) $k-1$ vectors. Since these $k-1$ vectors form a group of dimension $k-1$, to complete the subgroup $S$, this is equivalent to picking a dimension $1$ subgroup from a space of dimension $n-k+1$. Still, the count is not as simple as this, as the algorithm learns information from non-collisions as well. Namely, if they know $f(x)\neq f(y)$, then they also know $x-y\not\in S$, and thus have no reason to query this vector, or anything generated by it. Thus, we need to count the number of valid vectors that can still be used to complete the subspace. Since there are $\ell$ queries, the algorithm can know of at most $\ell$ distinct values of the function $f$. For each of these values, the algorithm can identify one unique vector to this value. Then, for each representative, by the above logic, we know that the pairwise difference cannot be an element in $S$. Furthermore, each of these vectors spans $p-1$ nonzero vectors, each of which also cannot be in $S$. Thus, the algorithm knows that to find the remaining vector in $S$, they cannot use the $p^{n-k+1}$ vectors already in $S$, nor the at most ${\ell\choose2}(p-1)\ell=\ell^2(p-1)$ vectors which are forbidden from taking the pairwise difference between bins. Then, the count of the number of possible subspaces is given by
    \begin{equation*}
        |\mathcal{H}|\geq \frac{p^{n}-p^{k-1}-\ell^2(p-1)}{p^k-p^{k-1}}.
    \end{equation*}
    Assume $\ell=c\cdot\sqrt{kp^{n-k}}$ for some constant $c$ (to be chosen later). Then, we have
    \begin{align*}
        |\mathcal{H}|&\geq \frac{p^{n}-p^{k-1}-c^2kp^{n-k}(p-1)}{p^k-p^{k-1}}
        =\frac{p^{n-k+1}-1-c^2kp^{n-2k+1}(p-1)}{p-1}
        =\sum_{i=0}^{n-k}p^i-c^2kp^{n-2k+1}.
    \end{align*}
    By choosing $c=1/\sqrt{k}$ (for constant $k$), we have that 
    $
        |\mathcal{H}|\geq p^{n-k}.
    $
    
    Then, supposing the query list only contains at most $k-1$ collisions, the query algorithm would essentially have to guess the correct subspace from among these, which happens with probability at most $1/p^{n-k}$. As such, no adaptive algorithm can succeed at identifying $S$ without using at least $\ell=\Omega(\sqrt{kp^{n-k}})$ queries.
\end{proof}

\section{About a previous version of this work}
    \label{app:previousversion}
An earlier version of this work~(\cite[Theorem 4.1]{CanonneCM26v1}, version 1) claimed a query lower bound of $\Omega(2^{n/2-o(1)})$ for the extremal Forrelation problem, using a construction based on partial spreads. We thank an anonymous reviewer for CCC 2026, who pointed out a gap in the argument. Specifically, the issue is that the prior approach was essentially ``greedily'' constructing a partial spread. Unfortunately, it is not known how to actually construct a partial spread greedily; whilst there are many algebraic constructions known, greedily extending an incomplete partial spread naively provably will lead to cases which cannot be extended to a complete partial spread, and whether a more clever construction allows the argument to go through is an open question. The current paper takes an entirely different approach to prove the lower bound.

KC also would like to thank William Kantor, who was contacted after the error, and generously answered many questions about partial spreads.
\end{document}